\documentclass[12pt]{article}
\usepackage[a4paper,margin=0.8in]{geometry}

\usepackage[round]{natbib}

\usepackage{array}
\newcolumntype{P}[1]{>{\raggedright\arraybackslash}p{#1}}

\usepackage{makecell}
\usepackage{appendix}
\usepackage{booktabs} 
\usepackage{pifont}
\usepackage[table]{xcolor}
\usepackage{amsmath, amssymb, amsthm}
\usepackage{algorithm}
\usepackage{algpseudocode}
\theoremstyle{plain}
\newtheorem{theorem}{Theorem}
\usepackage{tabularx}
\usepackage{newunicodechar}
\newunicodechar{★}{\ding{72}}
\usepackage[table]{xcolor}
\usepackage{xr-hyper}
\usepackage[colorlinks=true,linkcolor=blue,citecolor=blue,urlcolor=blue]{hyperref} 
\usepackage{cleveref}
\usepackage{caption}
\usepackage{subcaption}
\usepackage{graphicx}
\usepackage{enumitem}
\usepackage{dashrule}
\usepackage{authblk}

\usepackage{setspace}

\algrenewcommand\alglinenumber[1]{}

\definecolor{darkcherry}{HTML}{450314}
\definecolor{cherry}{HTML}{5A051A}
\definecolor{dark_purple}{HTML}{6E3D6F}
\definecolor{purple}{HTML}{7c397d}
\definecolor{red}{HTML}{D20A2E}
\definecolor{pink}{HTML}{EE1E52}
\definecolor{blue}{HTML}{1f77b4}
\definecolor{light_light_pink}{HTML}{F2BDCD}
\definecolor{grey}{HTML}{808080}
\definecolor{darkblue}{RGB}{31,119,180}

\newcommand{\trueposterior}{\tilde{\Pi}_\text{ref}}

\newcommand{\candidateposterior}{\tilde{\Pi}^\lambda}

\newcommand{\trueprior}{\Pi_\text{ref}}

\newcommand{\candidateprior}{\Pi^\lambda}

\newcommand{\datadim}{d_\mathcal{X}}
\newcommand{\paramdim}{d_\Theta}

\newcommand{\natparamdim}{d_\Lambda}
\newcommand{\natparamdimprior}{d_{{\Lambda}_{\pi}}}
\newcommand{\natparamdimloss}{d_{{\Lambda}_{L}}}

\newcommand{\estimatedFD}{\widehat{\mathrm{FD}}_m}

\newcommand{\FD}{\mathrm{FD}}
\newcommand{\estimatedFDposteriors}{\estimatedFD(\trueposterior \| \candidateposterior)}

\newtheorem{proposition}{Proposition}
\newtheorem{assumption}{Assumption}
\newtheorem{lemma}{Lemma}

\newtheorem{corollary}{Corollary}

\title{A computationally-tractable measure of global sensitivity for sampling-based Bayesian inference}
\author[1]{Arina Odnoblyudova\thanks{Corresponding author: arina.odnoblyudova.24@ucl.ac.uk}}
\author[1]{Charita Dellaporta}
\author[1]{Fran\c{c}ois-Xavier Briol}
\affil[1]{Department of Statistical Science, University College London}

\begin{document}

\date{\vspace{-6ex}}

\maketitle

% Using \doublespacing in the preamble 
% changes the text to double-line spacing
 \doublespacing

\begin{abstract}
    Bayesian inference can often be sensitive to the choice of hyperparameters of the prior or likelihood, yet defining and quantifying this sensitivity in a principled and computationally feasible way remains challenging in practice. Unfortunately, existing sensitivity methods are rarely applicable in modern Bayesian workflows due to their high computational cost and poor performance in moderate to high dimensions. To address these limitations, we introduce a new approach to global sensitivity analysis based on the Fisher divergence. Our method only requires a set of samples from a reference posterior and the ability to evaluate score functions, making it broadly computationally tractable. Under  regularity conditions, it controls changes in the whole posterior, and provides a bound on the impact of perturbations on the first two moments. We demonstrate these strengths on challenging Bayesian inference problems which are practically out of reach of existing approaches, including generalised Bayesian inference for unnormalised models, inference in Bayesian models of time series, and neural simulation-based inference. 
\end{abstract}


\section{INTRODUCTION}

Bayesian posteriors can be highly sensitive to modelling choices, particularly prior and likelihood hyperparameters. This sensitivity may substantially affect uncertainty quantification and predictive performance \citep{berger1994overview, ruggeri2005robust, Kleijn2006,doss2024scalable}, and may arise even when modelling choices vary within a plausible range. This challenge is further amplified in generalisations \citep{bissiri2016general, knoblauch2022generalized} or approximations \citep{Cranmer2020} of Bayesian inference, which frequently introduce additional hyperparameters whose influence is hard to assess. Consequently, evaluating sensitivity is a key component of the Bayesian workflow \citep{gelman2020bayesian}.

This task falls within the scope of Bayesian sensitivity analysis \citep{berger1994overview,ruggeri2005robust}, with existing methods generally categorised into local and global approaches, both of which analyse sensitivity relative to a reference prior/likelihood pair. The local approach \citep{ruggeri1993infinitesimal, gustafson1996local,roos2015sensitivity, al2021measuring, giordano2023evaluating, kallioinen2024detecting,doss2024scalable, di2025likelihood} employs differential techniques to assess the effect of infinitesimal perturbations around this reference. It is well studied and often computationally tractable, but its interpretation can be less intuitive and its relevance limited to small perturbations. 

In contrast, global Bayesian sensitivity analysis \citep{berger1986robust, ruggeri2005robust, kurtek2015bayesian, ghaderinezhad2022wasserstein, ho2023global}, the main focus of this paper, aligns more closely with practical concerns: it evaluates the maximum change in a quantity of interest within a neighbourhood of the reference choice. Practitioners often have a sense of a plausible range of alternative priors and likelihoods, and this approach directly quantifies how much their inferences could change within that range. 
Existing global sensitivity approaches can be categorised according to three main design choices: a measure used to quantify sensitivity, such as a posterior functional or the discrepancy to a reference posterior, a  neighbourhood in which sensitivity is measured, and a computational tool through which the method is implemented, typically combining an estimator for the sensitivity measure and a numerical optimiser. These algorithmic components interact, leading to significant trade-offs between three key desiderata: (i) \textit{the strength} of the sensitivity measure, i.e. whether the measure controls changes in the entire posterior or only some functionals such as the mean, (ii) the  \textit{interpretability} of the measure, and (iii) the \textit{computational tractability} of the method.

To date, no existing method has successfully balanced these desiderata in a way that is practically relevant for most modern Bayesian modelling problems, where the dimensionality of the posterior can be large, and posterior computation is expensive. The biggest limitation has been computational, with many existing methods either being limited to low dimensions \citep{moreno1991robust},  conjugate posteriors \citep{ruggeri2005robust}, requiring a large number of runs of a sampling method \citep{kurtek2015bayesian, ghaderinezhad2022wasserstein}, or requiring the solution of challenging  non-convex optimisation problems \citep{wasserman1993linearization, lavine2000linearization}.
Consequently, global sensitivity analysis is rarely applied in practice. Instead, many practitioners adopt an informal approach: they evaluate a few variants and assess their impact on the posterior. While straightforward, this can easily overlook other plausible modelling choices that may lead to substantially different inferences.

Our goal in this paper is to address the computational challenges that have limited the practical use of global sensitivity analysis. To this end, we propose a new sensitivity measure based on the \textit{Fisher divergence} (FD) \citep{hyvarinen2005estimation}. This measure is both strong and interpretable. Under mild regularity conditions, it is a statistical divergence, and under slightly stronger conditions, it controls sensitivity measures based on the difference in the first two moments, or constructed through the total variation, Kullback-Leibler and Wasserstein distance. The key advantage of the FD, however, is its computational tractability. It can be estimated using only a single set of samples from the reference posterior, without requiring repeated posterior sampling for each perturbed model in the sensitivity neighbourhood. This is achieved through evaluations of the score of these perturbed posteriors,  a key quantity used for many methods in Bayesian computation, including algorithms based on Langevin dynamics and Hamiltonian Monte Carlo \citep{Hoffman2014,Barp2018HMCReview,Fearnhead2024}. Furthermore, the FD has computational complexity which is linear in the number of samples and in the dimensionality of the posterior, and can be  approximated at a rate depending on the square root of the number of samples regardless of dimensionality. This makes it much better suited for measuring sensitivity in high-dimensional posteriors. Armed with this sensitivity measure, we consider two settings. The first, most general, consists of performing global optimisation over the space of hyperparameters, allowing us to consider a broad range of problems of practical relevance. The second, more limited, case considers neighbourhoods constructed through bounded convex polytopes of hyperparameters in exponential family priors. The latter setting does however allow us to recast the optimisation problem in a tractable form which can be solved exactly in a finite number of evaluations of the FD. 

We demonstrate the widespread applicability of our FD-based global sensitivity measure through extensive experiments on challenging problems including generalised Bayesian inference for an Ising model, Bayesian inference for autoregressive time-series, and amortised neural simulation based inference for a simulator in communications engineering.

\section{BACKGROUND}

\paragraph{Notation} We write $\mathcal{P}(\mathcal{Z})$ for the set of probability measures on some set $\mathcal{Z} \subseteq \mathbb{R}^{d_\mathcal{Z}}$. We denote by $L^2(P)$ the Lebesgue space of functions $f: \mathcal{Z} \rightarrow \mathbb{R}$ with respect to some measure $P \in \mathcal{P}(\mathcal{Z})$ such that $\int_{\mathcal{Z}} |f(z)|^2 \, d P (z) < \infty$. Moreover, we write $C^r(\mathcal{Z})$ for the set of $r$-times continuously differentiable functions $f: \mathcal{Z} \rightarrow \mathbb{R}$. For any $f \in C^1(\mathcal{Z})$ we denote its gradient by $\nabla f(z)=(\partial_1f(z), \dots, \partial_Df(z))^\top$, with $\partial_i=\partial/\partial z_i$. 
We use \textit{capital} letters (e.g. $P \in \mathcal{P}(\mathcal{Z})$) to denote probability measures and the corresponding \textit{lowercase} letters (e.g. $p: \mathcal{Z} \rightarrow \mathbb{R}^{+}$) to denote the associated probability density or probability mass functions. Finally, for any density $p: \mathcal{Z} \rightarrow \mathbb{R}^{+}$, we define the support $\text{supp}(p)\subseteq \mathcal{Z}$ as the domain where $p(z)>0$, and use $s_p := \nabla_z \log p(z)$ to denote the \textit{score function} of $p$. Throughout, we use $\lambda \in \Lambda \subseteq\mathbb{R}^{\natparamdim}$ to denote  hyperparameters of the inference procedure.

\subsection{Bayesian Inference and Sensitivity Analysis}

Consider some observed dataset $x_{1:n} = \{x_i\}_{i=1}^n$ consisting of (possibly dependent) realisations from some data-generating distribution $Q \in \mathcal{P}(\mathcal{X})$ for some continuous or discrete domain $\mathcal{X} \subseteq \mathbb{R}^{\datadim}$. We posit a parametric family $\{P_\theta^\lambda\}_{\theta \in \Theta} \subset \mathcal{P}(\mathcal{X})$ with corresponding likelihood $p_\theta(\cdot|\lambda):\mathcal{X}^n \to \mathbb{R}^+$, where $\theta \in \Theta \subseteq \mathbb{R}^{\paramdim}$ denotes the parameters of the model and will be assumed to be continuous, and $\lambda \in \Lambda$ are hyperparameters. To perform Bayesian inference, the modeller must select a prior distribution $\candidateprior \in \mathcal{P}(\Theta)$ which represents initial beliefs about $\theta$, and whose density will be written $\pi(\cdot|\lambda):\Theta \to \mathbb{R}^+$. Given this prior, inference is done through the posterior distribution with density $
    \tilde{\pi}_{\text{Bayes}}^\lambda (\theta|x_{1:n}) \propto   p_\theta(x_{1:n}|\lambda) \pi(\theta|\lambda)$, 
where $\propto$ denotes equality up to a multiplicative normalisation constant. Depending on the specific model and observed data, the posterior can be more or less impacted by changes in the hyperparameters $\lambda$. This is also the case for a plethora of approximations or generalisations of Bayesian inference, which we will refer to as posterior belief distributions. To cover all of these cases simultaneously, we will denote by $\candidateposterior \in \mathcal{P}(\Theta)$ any belief distribution constructed as a generalised, also called Gibbs, posterior \citep{bissiri2016general, knoblauch2022generalized} with density: 
\begin{align}\label{eq:gen_bayes}
    \tilde{\pi}^\lambda (\theta|x_{1:n}) \propto \exp\left(- L(\theta; x_{1:n},\lambda)\right) \pi(\theta|\lambda). 
\end{align}
where $L(\theta; x_{1:n},\lambda)$ is an empirical loss, and we will refer to these as `posteriors' for simplicity. A very common setting is $L(\theta; x_{1:n},\lambda) = \lambda l(\theta; x_{1:n})$ for some empirical loss $l: \Theta \times \mathcal{X}^n \to \mathbb{R}$, assumed to be bounded in $\theta$ from below, in which case the hyperparameter is called the learning rate $\lambda >0$ and controls the influence of the loss relative to the prior. 
Our paper will focus on several important instances of this general framework. This includes standard Bayes, which is recovered through $L(\theta; x_{1:n})= - \log p_\theta(x_{1:n})$. It also includes divergence-based posteriors \citep{jewson2018principles,matsubara2022robust,matsubara2024generalized}, where the loss is obtained through a statistical divergence between the model distribution and the empirical distribution of the data, and likelihood-free methods such as neural-likelihood estimation, which use a log-loss based on a neural surrogate of the likelihood \citep{Papamakarios2019}. However, the framework goes far beyond these examples, and also covers power posteriors \citep{grunwald2017inconsistency}, pseudo- or composite-likelihood posteriors \citep{Ribatet2012} and approximate Bayesian computation posteriors \citep{Schmon2020} amongst others.

Our aim in this paper is to measure the sensitivity of  $\candidateposterior$ to hyperparameters $\lambda \in \Lambda$. To formalise this, we follow the classical formulation of Bayesian global sensitivity analysis as reviewed by \citet{ ruggeri2005robust}. 
This requires the choice of a sensitivity measure, which is a functional 
$\rho: \mathcal{P}(\Theta) \to \mathbb{R}$ with respect to which we will assess changes. Examples include posterior expectations, i.e. $\rho(\tilde{\Pi}) := \mathbb{E}_{\theta \sim \tilde{\Pi}}[\tau(\theta)]$ for some $\tau: \Theta \to \mathbb{R}$  \citep{lavine1991sensitivity,ruggeri2005robust,ho2023global}, or the discrepancy between $\candidateposterior$ and some reference $\trueposterior\in \mathcal{P}(\Theta)$, i.e. $\rho(\tilde{\Pi}) := D(\trueposterior||\tilde{\Pi})$ where $D:\mathcal{P}(\Theta)\times \mathcal{P}(\Theta) \rightarrow \mathbb{R}$ is a statistical divergence   \citep{kurtek2015bayesian, ghaderinezhad2022wasserstein}.  
We also require a neighbourhood consisting of posteriors corresponding to plausible hyperparameters, which will be denoted $
    \mathcal{P}_{\Gamma} := \{\candidateposterior \in \mathcal{P}(\Theta) : \lambda \in  \Gamma \subseteq \Lambda\}$. 
The global Bayesian sensitivity over $\mathcal{P}_\Gamma$ is then defined as the largest possible change in 
$\rho(\candidateposterior)$ for hyperparameters in the neighbourhood: 
\begin{align}
\label{eq:posterior_imprecision}
S(\Gamma) := \sup_{ \tilde{\Pi} \in \mathcal{P}_\Gamma} \rho(\tilde{\Pi}) - \inf_{\tilde{\Pi} \in \mathcal{P}_\Gamma} \rho(\tilde{\Pi}) = \sup_{ \lambda \in \Gamma} \rho(\candidateposterior) - \inf_{\lambda \in \Gamma} \rho(\candidateposterior).
\end{align}
From a computational viewpoint, $S(\Gamma)$  must be estimated through a statistical estimator 
for $\rho(\tilde{\Pi})$, and a numerical optimiser to solve the supremum and infimum problems. 
Although a large number of methods have been proposed over the last four decades, these tend to suffer from the same three computational challenges: 
\begin{itemize}
 \item Firstly, for a given candidate posterior $\candidateposterior$ with fixed $\lambda \in \Lambda$, estimating $\rho(\candidateposterior)$ can be computationally challenging. For instance, when $\rho$ is defined through a statistical divergence, its estimation typically requires sampling from both  $\tilde{\Pi}_{\text{ref}}$ and $\candidateposterior$. If this estimator converges slowly in the number of samples or is computationally expensive, this will quickly become prohibitive. For example, \citep{ghaderinezhad2022wasserstein,mingo2026bayesian} quantify sensitivity through the Wasserstein distance, for which the most common estimator incurs a cost of $\mathcal{O}(m^3)$ and converges at a rate of $\mathcal{O}(m^{-1/\paramdim})$ in the number of samples $m$ when $\paramdim \geq 2$ \citep{Fournier2015}. 
 Similarly, \cite{kurtek2015bayesian} use the Fisher-Rao distance, for which common sample-based estimators would also suffer from a large computational cost and curse of dimensionality due to the need to perform (nonparametric) density estimation.
 For this reason, most Bayesian global sensitivity methods rely on simpler posterior functionals, such as posterior moments \citep{ruggeri2005robust,ho2023global} which have a cost of $\mathcal{O}(m)$ and convergence rate of $\mathcal{O}(m^{-1/2})$, but are not sufficient statistics for the entire posterior.

 \item  Secondly, beyond the difficulty of estimating $\rho$ once, most methods require estimating it repeatedly for candidates $\tilde{\Pi}^{\lambda_1},\tilde{\Pi}^{\lambda_2},\ldots$ for $\lambda_1,\lambda_2,\ldots \in \Lambda$ throughout the optimisation procedure. This often leads to the use of computationally expensive sampling schemes \citep{zhu2011bayesian,kurtek2015bayesian,giacomini2021robust,ghaderinezhad2022wasserstein}. One alternative to avoid this is through importance sampling using samples from $\trueposterior$, or to use a few iterations of sequential Monte Carlo; see e.g. \cite{ho2023global}. However, such approaches can lead to high-variance estimators of $\rho(\candidateposterior)$.

 \item  Thirdly, solving the optimisation problems in (\ref{eq:posterior_imprecision}) can be challenging, particularly when these are non-convex. 
Several approaches have been proposed; see e.g. the linearisation algorithm of \citep{lavine1991sensitivity} which approaches this trough a sequence of linear optimisation problems. However, most papers bypass the issue by taking a neighbourhood with a finite number of elements \citep{kurtek2015bayesian,ghaderinezhad2022wasserstein,ho2023global}, allowing for a direct solution to the supremum and infimum but leading to something more akin to the informal approach, thereby severely limiting their applicability.
\end{itemize}
Put together, these computational challenges mean that existing global sensitivity methods are not well-suited for modern Bayesian inference problems.

\subsection{The Fisher Divergence} \label{sec:fisherbackground}

This paper will tackle some of the aforementioned computational limitations by introducing a sensitivity measure based on the \textit{Fisher divergence} (FD; \cite{hyvarinen2005estimation}), also called the Hyv\"{a}rinen or score-matching divergence. The FD compares two distributions through their score functions. Given  $\Pi_1,\Pi_2 \in \mathcal{P}(\Theta)$ with corresponding scores $s_{\pi_1},s_{\pi_2}$, it is given by 
\begin{align}
    \FD(\Pi_1 \|\Pi_2) &:= \mathbb{E}_{\theta \sim \Pi_1} \left[ \| s_{\pi_1}(\theta) - s_{\pi_2}(\theta) \|_2^2 \right],
\end{align}
where $\| \cdot \|_2$ denotes the Euclidean norm. 
For the FD to be well-defined and a statistical divergence, that is, $\FD(\Pi_1\|\Pi_2) = 0 \iff \Pi_1 \equiv \Pi_2$ for all $\Pi_1,\Pi_2 \in \mathcal{P}(\Theta)$, certain regularity conditions are required, which we discuss in Section \ref{sec:method-FD}. 
As we will see in \Cref{subsec:interpretability}, the FD can be used to upper bound (under certain regularity conditions) the (1 or 2)-Wasserstein or total variation distances between $\Pi_1$ and $\Pi_2$,
making it a particularly strong notion of discrepancy \citep{huggins2018practical}. The FD is also computationally convenient as it can be expressed as an expectation and can therefore be estimated through samples from $\Pi_1$.

The FD has been used extensively, including for parameter estimation \citep{hyvarinen2005estimation,hyvarinen2007some,barp2019minimum,Yu2019,Scealy2022}, generalised Bayesian inference \citep{altamirano2023robust,Altamirano2024}, scoring rules \citep{Parry2012}, model comparison \citep{dawid2015bayesian,shao2019bayesian,jewson2022general}, and score-based generative modelling  \citep{Song2021}. 
In many of these settings, the FD is evaluated indirectly through an integration-by-parts identity that yields an equivalent objective up to an additive constant. 
The primary motivation is that it eliminates the need to evaluate the score of $\Pi_1$, which is typically unavailable or computationally intractable.
Our setting differs fundamentally from these applications. Because the posterior density is available in unnormalised form, its score can be evaluated directly and the FD can therefore be computed without resorting to integration by parts. Moreover, whereas most existing applications compute the FD between distributions on the data space $\mathcal{X}$, we apply it to distributions on the parameter space $\Theta$.

\section{METHODOLOGY}
We are now ready to propose our novel measure of global sensitivity based on the FD, explore optimisation routines, and discuss interpretability of the measure. 

\subsection{The Fisher Divergence as a Sensitivity Measure} \label{sec:method-FD}

Let $\trueposterior \in \mathcal{P}(\Theta)$
denote a reference posterior which is induced by a reference prior $\trueprior \in \mathcal{P}(\Theta)$ and loss function $L_{\text{ref}}:\Theta \times \mathcal{X}^n \rightarrow \mathbb{R}$ bounded from below so that $
\tilde{\pi}_{\text{ref}}(\theta|x_{1:n}) \propto \exp(-L_\text{ref}(\theta;x_{1:n}))\pi_{\text{ref}}(\theta)$.
Throughout, we will take $\trueposterior$ as corresponding to a default choice of prior and loss function made by the Bayesian modeller, and will aim to study sensitivity relative to it. We will also make the following assumption, which will ensure that the FD acts as a statistical divergence and will be sufficient to ensure our sensitivity measure is well-defined. 
\begin{assumption}\label{assump:fd} \;
\begin{itemize}
    \item 
    The reference posterior $\tilde{\Pi}_\text{ref}$ has support, $\text{supp}(\tilde{\pi}_{\text{ref}})$, on an open connected subset of $\Theta$.
    \item The functions $\pi_{\text{ref}}, L_\text{ref}(\cdot;x_{1:n}) \in C^1(\Theta)$ and $s_{\pi_{\text{ref}}}, \nabla_\theta L_\text{ref}(\cdot;x_{1:n}) \in L^2(\tilde{\Pi}_{\text{ref}})$. 
 
    \item The posterior neighbourhood $\mathcal{P}_{\Gamma}$ is a subset of $\mathcal{P}_\FD(\Theta)$, the space of distributions for which the FD to $\tilde{\Pi}_{\text{ref}}$ is well-defined: 
\begin{equation*}
\begin{aligned}
\mathcal{P}_\FD(\Theta) := \Big\{ &\tilde{\Pi} \in \mathcal{P}(\Theta) : \tilde{\pi}(\theta|x_{1:n}) \propto \exp(-L(\theta;x_{1:n}))\pi(\theta),   \text{supp}(\tilde{\pi}) = \text{supp}(\tilde{\pi}_{\text{ref}}), \\
& \ \quad  \pi,L(\cdot;x_{1:n}) \in C^1(\Theta)\text{ and } s_{\pi}, \nabla L(\cdot;x_{1:n}) \in L^2(\tilde{\Pi}_{\text{ref}}) \Big\}.
\end{aligned}
\end{equation*}
\end{itemize}
\end{assumption}
These assumptions are mild and will be satisfied in all our experiments. They can be further simplified if we only consider prior or learning rate sensitivity separately; see \Cref{appsec:div_derivation_quadratic_form}. The assumptions on the support are required because scores only measure rates of change of a density, rather than absolute mass, and two densities could have the same rate of change but very different mass if the support is disjointed. The differentiability and integrability conditions are needed to ensure the scores are well defined and square-integrable. 
Given these, we now propose to use the FD as our measure of sensitivity:
\begin{equation}\label{eq:FD_definition}
\begin{aligned}
 \rho^{\FD}(\tilde{\Pi}) := \FD(\tilde{\Pi}_{\text{ref}}\|\tilde{\Pi}) = \mathbb{E}_{\theta \sim \trueposterior}\left[ 
        \big\| s_{\tilde{\pi}_\text{ref}}(\theta) - s_{\tilde{\pi}}(\theta) \big\|_2^2 
        \right].
\end{aligned}
\end{equation}
As highlighted below, this is a strong measure of sensitivity since it controls changes in the entire posterior and can distinguish distributions in  $\mathcal{P}_{\Gamma} \subset \mathcal{P}_\FD(\Theta)$ from $\tilde{\Pi}_\text{ref}$.
\begin{proposition}[The FD distinguishes posteriors] \label{prop:fd-div}
    Suppose \Cref{assump:fd} holds.
    Then, for any $\tilde{\Pi} \in \mathcal{P}_\FD(\Theta)$, $\FD(\tilde{\Pi}_\text{ref}\|\tilde{\Pi}) = 0 \Leftrightarrow \tilde{\Pi}_{\text{ref}} \equiv \tilde{\Pi}$.   
\end{proposition}
The proof is in \Cref{appsubsec:fd-divergence}. We can now define the FD-based global sensitivity as
\begin{equation}\label{eq:FD_measure}
S^{\FD}(\Gamma)
:=
\sup\limits_{\lambda \in \Gamma} \rho^{\FD}(\candidateposterior)
-
\inf\limits_{\lambda \in \Gamma} \rho^{\FD}(\candidateposterior ).
\end{equation}
If the reference posterior is itself in the neighbourhood (i.e. $\tilde{\Pi}_{\text{ref}} \in \mathcal{P}_{\Gamma}$, or equivalently if $ \exists \lambda_{\text{ref}} \in \Gamma$ with $\tilde{\Pi}_{\text{ref}} = \tilde{\Pi}^{\lambda_{\text{ref}}}$) then the infimum term is $0$ since $\FD(\trueposterior||\trueposterior)=0$ following \Cref{prop:fd-div}. Throughout the remainder of this paper, we do not always assume $\tilde{\Pi}_{\text{ref}} \in \mathcal{P}_{\Gamma}$ as one may wish to study sensitivity to distributional properties, such as heavy tails or multimodality, that cannot conveniently be captured by a neighbourhood which includes $\trueposterior$. 

Before concluding, we make two important points relating to score functions. First, distances between scores depend on the parametrisation of $\Theta$, and so our approach is not invariant to reparametrisations. For simplicity, we propose to work with the same parametrisation as used for Markov chain Monte Carlo (MCMC) sampling, though we will see in our experiments that other choices are sometimes preferable. Second, using score functions for sensitivity analysis has a long history in Bayesian analysis. For example, \citet{West1984, Haro-Lopez1999} used scores to assess sensitivity of the posterior to individual data points. However, these methods focus mainly on local rather than global sensitivity, and do not use a divergence. 

\subsection{Estimation of the Sensitivity Measure}

Now that we have defined our measure of sensitivity, we discuss how to estimate it from samples. Since $\rho^\FD(\candidateposterior)$ is defined as an expectation under the reference $\tilde{\Pi}_\text{ref}$, it can be estimated through (possibly approximate) samples $\theta_1, \dots, \theta_m$ from $\tilde{\Pi}_\text{ref}$, leading to a natural estimator for $S^{\FD}(\Gamma)$:
\begin{align}\label{eq:est_sensitivity}
    \widehat{S}^{\FD}_m(\Gamma) &=
   \sup\limits_{\lambda \in \Gamma} \hat{\rho}^{\FD}_m(\candidateposterior) - \inf\limits_{\lambda \in \Gamma}
   \hat{\rho}^{\FD}_m(\candidateposterior), \\
   \text{with } \quad \hat{\rho}^{\FD}_m(\candidateposterior) & = \estimatedFD(\tilde{\Pi}_\text{ref}\|\candidateposterior) := \frac{1}{m} \sum_{i=1}^m \left\| s_{\tilde{\pi}_\text{ref}}(\theta_i) - s_{\tilde{\pi}^\lambda}(\theta_i) \right\|_2^2.
\end{align} 
While i.i.d. samples would lead to a standard Monte Carlo estimator, independent samples are rarely available in practice. Instead, inference will typically rely on samples generated by an MCMC algorithm targeting $\trueposterior$. 

The key advantage of this estimator is that it only relies on the score functions $s_{\tilde{\pi}_{\text{ref}}}$ and $s_{\tilde{\pi}^ \lambda}$, and therefore does not require knowledge of the corresponding normalisation constants of $\tilde{\pi}_{\text{ref}}$ and $\tilde{\pi}^ \lambda$, which are typically unavailable. This is in contrast with many other divergences used in sensitivity analysis, such as the Kullback-Leibler, Fisher-Rao \citep{kurtek2015bayesian} or Jensen-Shannon divergence \citep{kallioinen2024detecting}, who all require density evaluations. Assuming access to score functions is mild; in fact, many MCMC samplers based on Langevin dynamics or Hamiltonian Monte Carlo \citep{Hoffman2014,Barp2018HMCReview,Fearnhead2024} make use of scores to guide the Markov chain, and some of these quantities may therefore have already been pre-computed when obtaining the samples $\theta_1,\ldots,\theta_m$. 

Furthermore, another advantage of the FD is that we can estimate this divergence from a single set of samples from $\trueposterior$, without ever sampling from $\candidateposterior$. This is in contrast with virtually all of the existing divergence-based methods for sensitivity, which would require MCMC sampling for every new candidate under consideration or the use of importance sampling methods which may have high variance.

This estimator of the FD is also very computationally attractive:  the computational cost scales linearly in the number of samples $m$, with overall computational complexity $\mathcal{O}(m \paramdim)$. When the samples are obtained from a geometrically ergodic and reversible Markov chain, $\hat{\rho}^{\FD}_m(\candidateposterior)$ will satisfy a central limit theorem guaranteeing convergence to $\rho^{\FD}(\candidateposterior)$ at a rate of $\mathcal{O}(m^{-\frac{1}{2}})$ regardless of the dimension $\paramdim$ \citep[Theorem 4]{roberts2004general}. As we now show below, this result can be refined to a finite-sample complexity result:
\begin{assumption}\label{assumption:MCMC}
The samples $\{\theta_i\}_{i \in \mathbb{N}}$ are realisations from a Markov chain with invariant distribution $\trueposterior$, and there exists a Lyapunov function $V:\Theta \to [e,\infty)$ such that the chain is $V$-uniformly geometrically ergodic (A1 \& A2 in \citet{durmus2024probability}). Furthermore, the chain is initialised at some distribution $\xi$ and $\mathbb{E}_{\theta \sim \xi}[V(\theta)]<\infty$.
\end{assumption}
Many commonly used MCMC methods satisfy geometric ergodicity; see, Theorem 2.1 in \cite{livingstone2019geometric} and Theorem 9 in \cite{durmus2017convergence} for Hamiltonian Monte Carlo, and Theorem 16 from \cite{durmus2023convergence} for the No-U-Turn sampler. 
\begin{theorem}[Finite-sample complexity of the FD estimator under MCMC sampling]\label{prop:mcmc-complexity} 
Suppose Assumptions \ref{assump:fd} and \ref{assumption:MCMC} hold and fix some $\lambda \in \Gamma$. Suppose further that the squared score difference grows at most like $\sqrt{V(\theta)}$; i.e. $sup_{\theta\in\Theta} V(\theta)^{-\frac{1}{2}} \| s_{\tilde{\pi}_\mathrm{ref}}(\theta) - s_{\tilde{\pi}^\lambda}(\theta) \|_2^2 
        < \infty$.  
Then, there exists $C(\candidateposterior;\xi) < \infty$ independent of $m$ such that
$\forall \tau>0$,
\begin{align*}
    \mathbb{P}\left(
        \big| \estimatedFD(\tilde{\Pi}_\text{ref}\|\candidateposterior)
            - \FD(\tilde{\Pi}_\text{ref}\|\candidateposterior) \big|
        \geq \tau
    \right)
    \le \frac{C(\candidateposterior;\xi)}{m \tau^2}.
\end{align*}
\end{theorem}
The proof is in \Cref{appsubsec:mcmc-complexity} and the probability statement accounts for randomness in the Markov chain, including the distribution $\xi$ used for initialisation.  \Cref{prop:mcmc-complexity} shows that the estimated FD is close to the exact FD with high probability. Unlike most existing divergence-based sensitivity measures, we emphasise that the result has a rate in $m$ which is independent of $d_\Theta$. This will be essential to ensure that the method is widely applicable for Bayesian inference. We note that this a pointwise result, and does not directly imply a concentration inequality for $S^{\text{FD}}(\Gamma)$. For this, a stronger uniform will be required, and we will do this in \Cref{prop:mcmc-complexity-uniform} in the next section for a broad class of problems including prior sensitivity with exponential-family priors, or sensitivity analysis for the learning rate of generalised posteriors.

Although we now have an estimator $\widehat{S}^{\FD}_m(\Gamma)$ for $S^{\text{FD}}(\Gamma)$, this estimator still relies on an exact solution to two optimisation problems. This can be made explicit by writing
\begin{align}
\widehat{S}_{m}^\FD(\Gamma) =
   \sup\limits_{\lambda \in \Gamma} \hat{\rho}^{\FD}_m(\candidateposterior) - \inf\limits_{\lambda \in \Gamma}
   \hat{\rho}^{\FD}_m(\candidateposterior) = \hat{\rho}^\FD_m\left(\tilde{\Pi}^{\lambda^\text{sup}_{m}}\right) - \hat{\rho}^\FD_m \left(\tilde{\Pi}^{\lambda^\text{inf}_{m}}\right) 
\end{align}
where $\lambda^{\text{sup}}_m  \in \arg\sup_{\lambda \in \Gamma} \hat{\rho}^{\FD}_m(\candidateposterior)$ and $\lambda^{\text{inf}}_m \in \arg\inf_{\lambda \in \Gamma} \hat{\rho}^{\FD}_m(\candidateposterior)$.
The most direct approach is therefore to solve the two optimisation problems numerically with global optimisation algorithms. Given outputs $\lambda^\text{sup}_{m,t}$ and $\lambda^\text{inf}_{m,t}$ after $t \in \mathbb{N}$ iterations of such an optimiser, we can estimate sensitivity with $\widehat{S}_{m,t}^\FD(\Gamma) := \hat{\rho}^\FD_m (
    \tilde{\Pi}^{\lambda^\text{sup}_{m,t}}) - \hat{\rho}^\FD_m (\tilde{\Pi}^{\lambda^\text{inf}_{m,t}})$.
Direct optimisation will only be feasible practically when $d_\Lambda$ is relatively small and $\lambda \mapsto\candidateposterior$ is smooth, in which case most non-convex optimisation methods such as Bayesian optimisation 
or simulated annealing 
should perform well. However, the approach may not be able to scale to large $d_{\Lambda}$ due to the difficulties associated with high-dimensional non-convex optimisation.

\subsection{Convex Quadratic Formulation}

Interestingly, our next result shows that in a broad range of scenarios, we can make the optimisation problem much more tractable.
\begin{assumption}\label{assump:minimal-exp-fam}
The hyperparameter space $\Lambda = \Lambda_L \times \Lambda_\pi$, the loss is linear in $\lambda_L \in \lambda_L$, and the prior is in natural exponential family form with parameters $\lambda_\pi \in \Lambda_\pi$; i.e. for some loss $l(\cdot;x_{1:n}):\Theta \rightarrow\mathbb{R}^{d_{\Lambda_L}}$, sufficient statistic $T:\Theta \rightarrow \mathbb{R}^{d_{\Lambda_\pi}}$, and base density $g:\Theta \rightarrow [0,\infty)$:
\begin{align*}
    \mathcal{P}_\Gamma := \Big\{\candidateposterior:  &\; \lambda = [  \lambda_L, \lambda_\pi]^\top \in \Gamma = \Gamma_L \times \Gamma_\pi \text{ where } \Gamma_L \in \mathbb{R}^{\natparamdimloss}, \Gamma_\pi \in \mathbb{R}^{\natparamdimprior}, \\
    & L(\theta;x_{1:n},\lambda) = \lambda_L^\top l(\theta;x_{1:n}) \quad \text{and} \quad
       \pi(\theta|\lambda) \propto \exp\left(\lambda_\pi^\top T(\theta) + \log g(\theta)\right)\Big\}.
\end{align*}
\end{assumption}
\Cref{assump:minimal-exp-fam} covers a very broad range of problems. For example, linearity of the loss is satisfied when measuring sensitivity to the learning rate hyperparameter in generalised Bayesian inference. It is also satisfied when considering sensitivity to hyperparameters of a natural exponential family likelihood, or sensitivity to the weights in weighted log-likelihood approaches. Exponential family priors are extremely common and include Gaussian, Gamma, and Beta priors amongst many others. Under this assumption, the FD is a convex quadratic form.
\begin{proposition}[The FD as a convex quadratic form]\label{prop:FD_convex_quadratic}
    Suppose \Cref{assump:minimal-exp-fam} holds. Then, 
    \begin{align*}
\hat{\rho}^{\FD}_m(\candidateposterior)  =  \lambda^\top A \lambda + b^\top \lambda +  c,
    \end{align*}
    where for $J(\theta)
    :=
    [\,
    -\nabla_\theta l(\theta;x_{1:n}),
    \;
    \nabla_\theta T(\theta)^\top
    ]$, the matrix $A:=\frac{1}{m}\sum_{i=1}^m J(\theta_i)^\top J(\theta_i)$, the vector $b :=
        -\frac{2}{m}
        \sum_{i=1}^m
        J(\theta_i)^\top
        (
        s_{\tilde{\pi}_{\mathrm{ref}}}(\theta_i)
        -
        s_g(\theta_i)
        )$, and $c
:=
\frac{1}{m}
\sum_{i=1}^m
\|
s_{\tilde{\pi}_{\mathrm{ref}}}(\theta_i)
-
s_g(\theta_i)
\|_2^2$.
    Furthermore, $A$ is positive semi-definite, and thus the quadratic form is convex in $\lambda \in \Lambda$. 
\end{proposition}
See \Cref{app:proof_quadratic_form} for the proof.
The matrix $A=(A_{ij})_{i,j\in\{L,\pi\}}$ has blocks given by
    $
    A_{LL}
    := \frac{1}{m} \sum_{i=1}^m
          \nabla_\theta l(\theta_i;x_{1:n})^\top  \nabla_\theta l(\theta_i;x_{1:n})$,
    $A_{L \pi}
    := A_{\pi L}^\top
       = -\frac{1}{m}  \sum_{i=1}^m
       \nabla_\theta l(\theta_i;x_{1:n})^\top
       \nabla_\theta T(\theta_i),
    $  
        $A_{\pi\pi}
    := \frac{1}{m}  \sum_{i=1}^m
       \nabla_\theta T(\theta_i)^\top\nabla_\theta T(\theta_i)$. 
The sensitivity measure also remains a convex quadratic form if we consider only sensitivity to either the prior or the loss; see \Cref{appsec:div_derivation_quadratic_form}. The complexity of computing the quadratic form is $\mathcal{O}\bigl(m \natparamdim^2 \paramdim\bigr)$ which has a favourable scaling in $m$ and $\paramdim$. However, the main computational advantage of this convex objective is that, depending on the geometry of the neighbourhood $\Gamma$, both optimisation problems can be solved efficiently. This is the case for bounded convex polytopes \citep{Henk2017}, such as hyperrectangles, simplices or $l_1$-balls, which can be represented through the convex hull of a finite set of vertices:
\begin{align}\label{eq:polytope-constrained-neighbourhood}
    \Gamma_\text{poly} := \left\{ \lambda \in \Lambda : \lambda = \sum_{k=1}^K \gamma_k v_k  \text{ for } v_1,\ldots,v_K \in \Lambda, \gamma_1,\ldots,\gamma_K \geq 0, \sum_{k=1}^K \gamma_k = 1\right\}.
\end{align}
Since $\Gamma_\text{poly}$ is convex and compact, we know by Bauer’s maximum principle that the supremum of $\hat{\rho}^\FD_m(\tilde{\Pi}^\lambda)$ will be achieved at an extreme point \citep[see][Theorem 7.42]{beck2014introduction}. Thus, we can solve the supremum problem with an $\mathcal{O}(K)$ cost by evaluating the FD at the vertices $v_1,\ldots,v_K$ through simple enumeration. 
Furthermore, the infimum problem is a convex minimisation problem over a convex and compact set, which can be solved by 
projected gradient descent.
This carries a cost driven by the maximum of  $\mathcal{O}(\natparamdim^2)$, the cost of performing a gradient step, and the cost of the projection on $\Gamma_{\text{poly}}$, and converges 
at a rate $
\mathcal{O}( \kappa^t)$ for some $\kappa \in (0,1)$ when $A$ is strongly convex (see Section 7.3.3 from \citet{wright2022optimization}).

A special case of a bounded convex polytope is obtained through box-type constraints (also called hyperrectangles), where the boundaries for each hyperparameter are expressed through the vectors $\lambda',\lambda'' \in \Lambda$ so that $\Gamma_\text{box} := \{ \lambda \in \Lambda : \lambda'_j \leq \lambda_j \leq \lambda''_j \text{ for } j \in \{1,\ldots,\natparamdim\} \}$. Since $\Gamma_\text{box}$ has at most $K=2^{\natparamdim}$ vertices, computing the supremum for box-constrained sensitivity analysis is computationally inexpensive  for small to moderate $\natparamdim$, but may be infeasible for large $\natparamdim$. When computing the infimum through projected gradient descent, projections reduce to coordinate-wise clipping, which carries a cost of $\mathcal{O}(\natparamdim)$. 

Beyond computational tractability,  \Cref{assump:minimal-exp-fam} also allows us to  extend the pointwise result of \Cref{prop:mcmc-complexity} to get a root-$m$ finite sample complexity result for $S^{\FD}(\Gamma)$.
\begin{theorem}[Finite-sample complexity for the global sensitivity measure]\label{prop:mcmc-complexity-uniform}
Suppose Assumptions \ref{assump:fd}, \ref{assumption:MCMC} and \ref{assump:minimal-exp-fam} hold, $\Gamma$ is compact and the following growth condition holds
\begin{equation*}
\sup_{\theta\in\Theta}V(\theta)^{-1/4}
\max\left\{
\|\nabla_\theta l(\theta;x_{1:n})\|_{F},
\|\nabla_\theta T(\theta)\|_{F},
\|s_{\widetilde\pi_{\mathrm{ref}}}(\theta)\|_2,
\|s_g(\theta)\|_2
\right\} 
<\infty.
\end{equation*}
Then, there exists $0 < C_{\Gamma,\xi} < \infty$ independent of $m$ such that, for any $\delta \in (0,1)$, we have with probability $1-\delta$ that
\begin{align*}
     \widehat{S}^{\FD}_m(\Gamma) - \frac{C_{\Gamma,\xi}}{\sqrt{m \delta}} \leq S^{\FD}(\Gamma) \leq  \widehat{S}^{\FD}_m(\Gamma) +  \frac{C_{\Gamma,\xi}}{\sqrt{m\delta}}.
\end{align*}
\end{theorem}
The proof is in \Cref{appsubsec:mcmc-complexity-uniform}. The theorem guarantees that the difference between estimated and exact sensitivity measures decreases as $O(m^{-\frac{1}{2}})$, guaranteeing that neither over- nor under-estimation of sensitivity is large. The upper bound is particularly helpful since it guarantees we are unlikely to significantly underestimate global sensitivity. Here, the growth condition is stated in terms of the gradient of the loss and sufficient statistic, as well as the scores of the reference posterior and $g$, which makes the condition easier to check.

Before concluding, we emphasise that the choice of parametrisation for $\Lambda$ is important here. For example, although natural exponential family priors lead to a convex quadratic form, this may not be the case for other parametrisations of an exponential family model. For exponential families in minimal form, this is not a problem as there always exists a smooth and invertible mapping between parametrisations (see e.g. Section 8.1 in \citet{barndorff2014information}). However, specifying a neighbourhood in natural parameter space may be challenging, and it is not always possible to guarantee that a neighbourhood specified in the original parametrisation will remain a bounded convex polytope once mapped to the natural parametrisation. One exception is affine transformations, which maintain the required geometry of the domain.

\section{STRENGTH AND INTERPRETABILITY}\label{subsec:interpretability_of_results}\label{subsec:interpretability}

We have now proposed a sensitivity measure which is computationally tractable in that it can be estimated without suffering from a curse of dimensionality, and where the optimisation problem can be solved efficiently. The remaining questions are whether the measure is meaningful and whether it is interpretable. 

\subsection{Strength of the Measure of Sensitivity}

We have seen in \Cref{prop:fd-div} that $S^{\text{FD}}(\Gamma)$ is based on a divergence, which makes it a stronger measure than those based on finitely many moments. Interestingly, it is also possible to relate $S^{\text{FD}}(\Gamma)$ to measures based on alternative divergences given slightly stronger assumptions on $\Gamma$.
\begin{assumption}
\label{assump:candidate-prior-regularity} \;
\begin{enumerate}
    \item[(A4.1)] For all $\lambda \in \Gamma$,
$\log \tilde{\pi}^\lambda$ has strongly concave tails, i.e. $\exists K_\lambda, R_\lambda \geq 0$, such that $- \nabla_\theta^2 \log \tilde{\pi}^\lambda(\theta) \succeq K_\lambda I_{\paramdim}$ for all $\theta \in \text{supp}(\tilde{\pi}_{\text{ref}})$ with $\|\theta\|_2 \ge R_\lambda$. Additionally, $K_{\Gamma} := \inf_{\lambda \in \Gamma} K_\lambda > 0$, $R_{\Gamma} := \sup_{\lambda \in \Gamma} R_\lambda < \infty$.
    \item[(A4.2)] The candidates are uniformly bounded from above and below near the center of $\Theta$; i.e. for $\tilde{\pi}_{\text{min}}^{R_\Gamma} := \inf_{\lambda \in \Gamma} \inf_{\|\theta\|_2 \leq 3R_\lambda} \tilde{\pi}^\lambda(\theta)$ and $ \tilde{\pi}_{\text{max}}^{R_\Gamma} := \sup_{\lambda \in \Gamma} \sup_{\|\theta\|_2 \leq 3R_\lambda} \tilde{\pi}^\lambda(\theta)$, we have $0 < \tilde{\pi}^{R_\Gamma}_{\text{min}} \leq \tilde{\pi}^{R_\Gamma}_{\text{max}} < \infty$.
\end{enumerate}
\end{assumption}
The first assumption corresponds to requiring that candidates are uniformly strongly log-concave \citep{Saumard2014} outside of some region near the origin. This holds if both $L$ and $\Pi^\lambda$ are strongly log-concave outside a compact region. This is satisfied by many light-tailed models, including Gaussians and more generally densities whose negative log-density has Hessian bounded below by a positive constant outside a compact set. However, it excludes some commonly used distributions including the Student-t, Laplace and Gamma distributions, whose tail curvature vanishes asymptotically.
The second assumption  requires the candidate posteriors to be bounded above and below on the centre of the distribution, which is relatively mild and will hold when the candidate priors and likelihoods are also bounded above and below over that region. Under these assumptions, and assuming for simplicity that $\trueposterior \in \mathcal{P}_{\Gamma}$, the FD sensitivity is strong in that it controls sensitivity in three widely studied divergences. 
\begin{proposition}\label{prop:dominate_otherdivergences}
   Suppose that $\lambda_{\text{ref}} \in \Gamma$, Assumptions \ref{assump:fd} and \ref{assump:candidate-prior-regularity} hold, 
    and let $\text{W}_p$, $\text{TV}$ and $\text{KL}$ denote
    the $p$-Wasserstein, total variation distance and Kullback-Leibler divergence respectively. 
    Then, $\exists \alpha_{\Gamma} > 0$ depending on $K_{\Gamma}$, $R_\Gamma$, $\tilde{\pi}^{R_\Gamma}_{\text{min}}$ and $\tilde{\pi}^{R_\Gamma}_{\text{max}}$ such that:
    \begin{align*}
        \sup_{\lambda \in \Gamma} W_p(\trueposterior, \candidateposterior) := S^{\text{W}_p}(\Gamma) \leq \frac{1}{\alpha_{\Gamma}} \sqrt{S^{\text{FD}}(\Gamma)} \quad \text{and}\quad
        \sup_{\lambda \in \Gamma} \text{TV}(\trueposterior \| \candidateposterior) := S^{\text{TV}}(\Gamma) \leq \frac{1}{\sqrt{2\alpha_{\Gamma}}} \sqrt{S^{\text{FD}}(\Gamma)}.
    \end{align*}
    for $p \in \{1,2\}$. Moreover, if for all $\lambda \in \Lambda$, $ \exists K^\prime_\lambda > 0$ such that $ \tilde{\pi}^\lambda$ is $K^\prime_\lambda$-strongly log-concave over the whole domain (i.e. $- \nabla_\theta^2 \log \tilde{\pi}^\lambda(\theta) \succeq K_\lambda I_{\paramdim}$  $\forall \theta \in \Theta$) and $\alpha^\prime_{\Gamma} := \inf_{\lambda \in \Gamma} K^\prime_{\lambda} > 0$, then
    \begin{align*}
        \sup_{\lambda \in \Gamma} \text{KL}(\trueposterior \| \candidateposterior) := S^{\text{KL}}(\Gamma) \leq \frac{1}{2\alpha^\prime_{\Gamma}} S^{\text{FD}}(\Gamma).
    \end{align*}
\end{proposition}
The proof is in \Cref{app:proof_dominate_divergences}. We note that the conditions for controlling the KL divergence are much stronger since we require strong log concavity over the entire domain, rather than just in the tails.
Interestingly, under \Cref{assump:candidate-prior-regularity}, $S^{\FD}(\Gamma)$ also becomes much more interpretable, in the sense that it controls differences in the first two moments of $\trueposterior$ and $\candidateposterior$. This is a direct corollary of the result involving the Wasserstein distance; see \Cref{appsubsec:interpretability-through-moment-control} for the proof.
\begin{corollary}\label{prop:moment-bounds}
Suppose that $\lambda_{\text{ref}} \in \Gamma$, Assumptions \ref{assump:fd} and \ref{assump:candidate-prior-regularity} hold, and for any $\tilde{\Pi} \in \mathcal{P}_{\text{FD}}(\Theta)$ with finite second moments, denote by $\mu_{\tilde{\Pi}}, \Sigma_{\tilde{\Pi}}$ the corresponding mean vector and covariance matrix. 
Then $\exists \alpha_{\Gamma} > 0$ depending on $K_{\Gamma}, R_{\Gamma}, \tilde{\pi}^{R_\Gamma}_{\text{min}}$ and $\tilde{\pi}^{R_\Gamma}_{\text{max}}$ such that:
\begin{align*}
S^{\mathrm{mean}}(\Gamma)
& :=
\sup_{\lambda \in \Gamma}
\left\|\mu_{\trueposterior} -\mu_{\candidateposterior} \right\|_2
\le 
\frac{1}{\alpha_{\Gamma}}
\sqrt{S^\mathrm{FD}(\Gamma)}, \quad \text{and } \\
S^{\mathrm{cov}}(\Gamma) &
:=
\sup_{\lambda \in \Gamma}
\|\Sigma_{\trueposterior}-\Sigma_{\candidateposterior}\|_2
\le
\frac{3}{\alpha_{\Gamma}}\,
\min\left(
\sqrt{\|\Sigma_{\trueposterior}\|_2},
\ \sup_{\lambda\in\Gamma}\sqrt{\|\Sigma_{\candidateposterior}\|_2}
\right)
\sqrt{S^{\FD}(\Gamma)}
+\frac{5.25}{\alpha_{\Gamma}^2}\,
S^{\FD}(\Gamma).
\end{align*}
\end{corollary}

\subsection{Interpretability of the Measure of Sensitivity}\label{subsubsec:interpretability-through-factorisation}

Before concluding, we also consider the interpretability of $S^{\text{FD}}(\Gamma)$.
As for all divergence-based global sensitivity measures, it is hard to interpret raw values 
, and we therefore propose to focus on relative comparisons. Given two neighbourhoods $\Gamma_1,\Gamma_2 \subset \Lambda$, we can use $S^{\text{FD}}(\Gamma_1)-S^{\text{FD}}(\Gamma_2)$ to measure the relative sensitivity to $\Gamma_1$ and $\Gamma_2$. In addition, for a single neighbourhood, the sensitivity measure typically simplifies through various decompositions which can shed more light on what is being measured, including the relative sensitivity to the prior and loss function, or the relative sensitivity to different hyperparameters. 

\paragraph{Decomposition 1: Impact of the loss/prior} We can expand the squares in \Cref{eq:FD_definition}:
\begin{align}\label{eq:FD_decomposition}
    \rho^{\FD}(\tilde{\Pi}) & = \mathbb{E}_{\theta \sim \tilde{\Pi}_{\text{ref}}}\left[\left\| \nabla_{\theta} L_{\text{ref}}(\theta;x_{1:n}) - \nabla_{\theta} L(\theta;x_{1:n})\right\|^2_2 \right] + \mathbb{E}_{\theta \sim \tilde{\Pi}_{\text{ref}}}\left[\left\| s_{\pi_{\text{ref}}}(\theta) - s_{\pi}(\theta) \right\|^2_2 \right] \nonumber \\
    & \quad + 2\mathbb{E}_{\theta \sim \tilde{\Pi}_{\text{ref}}}\left[ \left( \nabla_{\theta} L_{\text{ref}}(\theta;x_{1:n}) - \nabla_{\theta} L(\theta;x_{1:n})\right)^\top \left(s_{\pi_{\text{ref}}}(\theta) - s_{\pi}(\theta)\right)\right]
\end{align}
This decomposition is particularly interpretable: the first term is non-negative and measures changes in the loss, the second term is also non-negative and measures changes in the prior, whereas the third term could be positive or negative and measures interactions between differences in the loss and prior. These three terms could be estimated separately through MCMC and used, for example, to compare the relative impact of changes in $\lambda$ on the loss and prior. 

\Cref{eq:FD_decomposition} is particularly informative in several special cases. For example, if we assume the reference posterior is an element of the neighbourhood (i.e. $\tilde{\Pi}_{\text{ref}} = \tilde{\Pi}^{\lambda_{\text{ref}}}$ for some $\lambda_{\text{ref}} \in \Gamma$) and consider only hyperparameters $\lambda$ of the loss, then the last two terms are zero and $
    S^\FD(\Gamma) = \sup_{\lambda \in \Gamma} \mathbb{E}_{\theta \sim \tilde{\Pi}_{\text{ref}}}[\| \nabla_{\theta} L(\theta;x_{1:n},\lambda_{\text{ref}}) - \nabla_{\theta} L(\theta;x_{1:n},\lambda)\|^2_2]$.
A similar expression holds when considering only prior hyperparameters $\lambda$: $
    S^\FD(\Gamma) = \sup_{\lambda \in \Gamma} \mathbb{E}_{\theta \sim \tilde{\Pi}_{\text{ref}}}[\| s_{\pi(\cdot|\lambda_{\text{ref}})}(\theta) - s_{\pi(\cdot|\lambda)}(\theta)\|^2_2 ]$.
In both special cases, we note that the expression can be further reduced to a simplified quadratic form under \Cref{assump:minimal-exp-fam}; see \Cref{appsec:div_derivation_quadratic_form} for more details.

\paragraph{Decomposition 2: Independence across dimensions} Suppose that the reference and candidate posteriors factorise due to independence across subsets of dimensions; that is, $\Theta = \Theta_1 \times \ldots \times \Theta_J$, $\tilde{\pi}_{\mathrm{ref}}(\theta)= \prod_{j=1}^J \tilde{\pi}_{\mathrm{ref},j}(\theta_j)$ and $\tilde{\pi}(\theta)= \prod_{j=1}^J \tilde{\pi}_{j}(\theta_j)$. Then, we have $
\rho^{\FD}(\tilde{\Pi})
= \FD(\tilde{\Pi}_{\text{ref}}\|\tilde\Pi) = 
\sum_{j=1}^J
\FD(\tilde{\Pi}_{\text{ref},j} \| \tilde{\Pi}_j)$.
If the neighbourhood can also be written as $\Gamma = \Gamma_1 \times \ldots \times \Gamma_J$ with hyperparameter $\lambda_j \in \Gamma_j$ only entering through the $j^{\text{th}}$ factor $\tilde{\pi}^{\lambda_j}_j(\theta_j)$ so that $\tilde{\pi}^\lambda(\theta)= \prod_{j=1}^J \tilde{\pi}_{j}^{\lambda_j}(\theta_j)$, then we obtain a sum of FD measures over lower-dimensional spaces:
\begin{align}\label{eq:optimisation-decomposition-composite-prior}
S^{\FD}(\Gamma)
=
\sum_{j=1}^J
\sup_{\lambda_j \in \Gamma_j} \FD(\tilde{\Pi}_{\text{ref},j} \| \tilde{\Pi}_j)
-
\inf_{\lambda_j \in \Gamma_j}
\FD(\tilde{\Pi}_{\text{ref},j} \| \tilde{\Pi}_j)
=: \sum_{j=1}^J S^{\FD}(\Gamma_j),
\end{align}
which allows us to examine individual terms to assess how much sensitivity depends on each of the subsets of dimensions. This posterior-factorisation assumption is relatively strong, since dependence across parameters is typically induced by the likelihood. However, a related broadly applicable special case arises when we consider prior sensitivity and only the prior factorises. Suppose $\pi_{\text{ref}}(\theta) = \pi(\theta\mid\lambda_{\text{ref}})$ and $\pi(\theta\mid\lambda)= \prod_{j=1}^J \pi_{j}(\theta_j\mid\lambda_j)$, but $\trueposterior$ and $\candidateposterior$ need not factorise. When considering prior sensitivity, the loss and interaction terms in the FD decomposition vanish and the sensitivity measure depends solely on differences between prior scores. Crucially, each summand depends on $\theta$ only through its $j$-th coordinate $\theta_j$. Since expectations of a function of $\theta_j$ alone reduce to expectations under the corresponding marginal of $\trueposterior$, this is equivalently $
    S^\FD(\Gamma) = \sum_{j=1}^J \sup_{\lambda_j \in \Gamma_j} \mathbb{E}_{\theta_j \sim \tilde{\Pi}_{\text{ref},j}}[\| s_{\pi_j(\cdot|\lambda_{\text{ref},j})}(\theta_j) - s_{\pi_j(\cdot|\lambda_j)}(\theta_j)\|^2_2 ]$,
where $\tilde{\Pi}_{\text{ref},j}$ denotes the marginal law of $\theta_j$ induced by $\trueposterior$.
This decomposition is not only interesting for interpretability, but also for computation. Indeed, it allows us to compute each term of this sum separately, which can significantly lower the dimensionality of the optimisation problems.

\section{EXPERIMENTS}\label{sec:experiments}
We are now ready to evaluate our sensitivity measure, and our code can be found at \url{https://github.com/jularina/fd-sens}. All experiments were run on macOS (Apple Silicon, M1 Pro) using a single CPU core without GPU acceleration. We use the packages \texttt{sbi} \citep{tejero2020sbi}, \texttt{POT} \citep{flamary2021pot}, and \texttt{scipy.optimize} \citep{virtanen2020scipy} for posterior inference, optimal transport computations, and numerical optimisation. 
We also use \texttt{PosteriorDB} \citep{magnusson2025posteriordb} to access benchmark Bayesian models. Additional proofs of the numerical experiment assumptions are provided in \Cref{appsec:complementary_experiments}.

\subsection{Synthetic Conjugate Gaussian Location Models} \label{subsec:toy_gaussian_model}
We first study the sensitivity of synthetic conjugate Gaussian location models to prior hyperparameters. This toy model allows closed-form verification of our method and benchmark comparison with existing approaches. 

\paragraph{Description of the model.}
We consider a Gaussian location model $\mathcal{N}(\theta,\Sigma_l)$ with mean parameter 
$\theta \in \mathbb{R}^{\paramdim}$ and covariance matrix $\Sigma_l$ assumed known.
We place a Gaussian prior $\candidateprior=\mathcal{N}(\mu,\Sigma)$ on $\theta$,  which can be expressed in natural exponential family parametrisation as $\pi(\theta | \lambda)
=
\exp\!\left(
\lambda^\top T(\theta)
-
A(\lambda)
\right)$
with $\lambda := [\lambda_0, \Lambda_1]=[\Sigma^{-1}\mu, -\tfrac12 \Sigma^{-1}]$ and $T(\theta) := [\theta, \theta \theta^\top]$. Given independent observations $x_1,\ldots,x_n$, the posterior is conjugate Gaussian $\candidateposterior=\mathcal{N}(\mu_n,\Sigma_n)$ with parameters $
\Sigma_n
=
(
-2\Lambda_1
+
n\Sigma_l^{-1}
)^{-1}$ and $
\mu_n
=
\Sigma_n
(
\lambda_0
+
n\Sigma_l^{-1}\bar{x}
)$,
where $\bar{x}=n^{-1}\sum_{i=1}^n x_i$. In this example, several divergences can be computed in closed-forms: 
the FD between 
$\tilde{\Pi}^{\lambda_1}=\mathcal{N}(\mu_{1,n},\Sigma_{1,n})$ and
$\tilde{\Pi}^{\lambda_2}=\mathcal{N}(\mu_{2,n},\Sigma_{2,n})$
is $
\FD(
\tilde{\Pi}^{\lambda_1}
\|
\tilde{\Pi}^{\lambda_2})
=
\|
\Sigma_{2,n}^{-1}
(
\mu_{2,n}-\mu_{1,n}
)
\|_2^2
+
\operatorname{tr}
\! \;(
(
\Sigma_{2,n}^{-1}
-
\Sigma_{1,n}^{-1}
)^2
\Sigma_{1,n}
)$, whilst the KL divergence and Wasserstein distance expressions are given in
\Cref{appsubsubsec:toy_gaussian_model_closed_forms}.

\paragraph{Illustrative univariate experiment.}
We begin with a dataset $\{x_i\}_{i=1}^{100}$ generated from
$\mathcal{N}(\theta_\text{true},\sigma_l^2)$ with unknown 
$\theta_\text{true}=3$ and known $\sigma_l=2$. Our reference prior $\trueprior$ is 
$\mathcal{N}(\mu_{\text{ref}},\sigma_{\text{ref}}^2)$
with $\mu_{\text{ref}}=2$ and $\sigma_{\text{ref}}=4$. We draw $m=2000$ independent samples from $\trueposterior$ and study sensitivity to $\mu$ and $\sigma$; see \Cref{fig:toy_gaussian_model_a}. As expected, $\hat{\rho}^\text{FD}_m(\candidateposterior)$ increases the further away the prior mean gets from $\mu_{\text{ref}}=2$. Interestingly, as the prior becomes less informative (i.e., as $\sigma$ increases), the effect of deviations in the prior mean $\mu$ diminishes, showing that the FD correctly identifies that a non-informative prior will be less sensitive to its location.

\begin{figure*}[t!]
    \centering
    \begin{subfigure}[b]{0.32\textwidth}
        \centering
        \includegraphics[width=\linewidth]{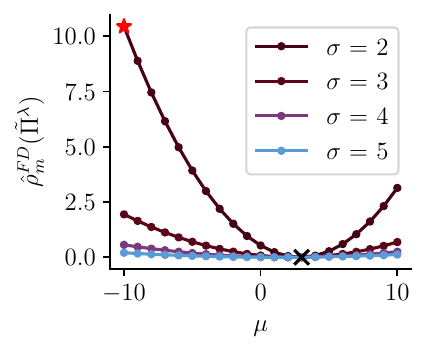}
        \caption{}
        \label{fig:toy_gaussian_model_a}
    \end{subfigure}
    \hfill
    \begin{subfigure}[b]{0.33\textwidth}
        \centering
        \includegraphics[width=\linewidth]{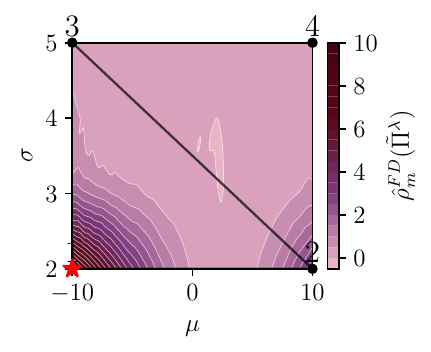}
        \caption{}
        \label{fig:hyperparams_contour_from_corners}
    \end{subfigure}
    \hfill
    \begin{subfigure}[b]{0.32\textwidth}
        \centering
        \includegraphics[width=\linewidth]{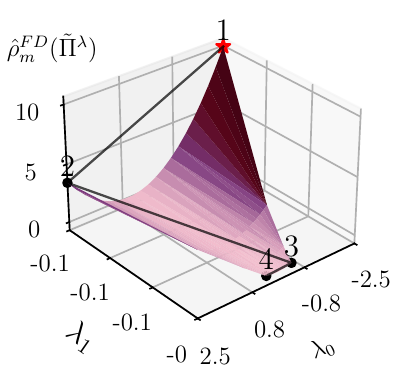}
        \caption{}        \label{fig:toy_gaussian_model_optimisation_nonconvex}
    \end{subfigure}
    \caption{\emph{Synthetic Gaussian location model.}
    \textcolor{red}{★} denotes the optimal candidate identified via box-constrained optimisation that maximises $\hat{\rho}_m^{\FD}(\tilde{\Pi}^\lambda)$.
    (a) The FD as a function of the prior hyperparameters $\mu$ and $\sigma$. The symbol \textbf{$\times$} marks the minimum.
    (b) The FD is non-convex in the $\mu,\sigma$ parametrisations. (c) In contrast, reparameterising the FD in terms of the natural parameters $\lambda_0 = \mu/\sigma^2$ and $\lambda_1 = -0.5/\sigma^2$ yields a convex function. Corner points are numbered, and the trajectory indicates the optimisation exploration path.
    }
    \label{fig:toy_gaussian_model}
\end{figure*}
The remainder of the figure illustrates the importance of the choice of parametrisation for optimisation. Over the Gaussian prior neighbourhood
$\{(\mu,\sigma)^\top \in [-10,10] \times [2,5]\}$,  the FD is non-convex in $(\mu,\sigma)$  (see  \Cref{fig:hyperparams_contour_from_corners}). However, we recall that \Cref{prop:FD_convex_quadratic} guarantees that the FD is a convex quadratic for natural exponential family models. We can therefore reparameterise the Gaussian prior in terms of natural parameter $\lambda_0 = \frac{\mu}{\sigma^2}$ and $\lambda_1 = -\frac{1}{2\sigma^2}$, in which case  the objective is indeed quadratic (see \Cref{fig:toy_gaussian_model_optimisation_nonconvex}). 
In this new parametrisation, the same neighbourhood can be expressed as $\Gamma = \{ (\lambda_0,\lambda_1)^\top: \lambda_1 \in [-\frac{1}{8},-\frac{1}{50}] \text{ and } |\lambda_0| \leq 20 |\lambda_1| \}$, 
which is a bounded convex polytope, and 
computation of the supremum boils down to maximum of the FD at the four vertices, which is attained at $\mu=-10$ and $\sigma=2$. This example therefore illustrates the importance of working in natural exponential family parametrisations, and we will hence focus on this case in the remainder of the paper.

\paragraph{Strength of the FD measure.}

\begin{figure}[t!]
    \centering
    \begin{subfigure}[b]{0.32\linewidth}
        \centering
        \includegraphics[width=\linewidth]{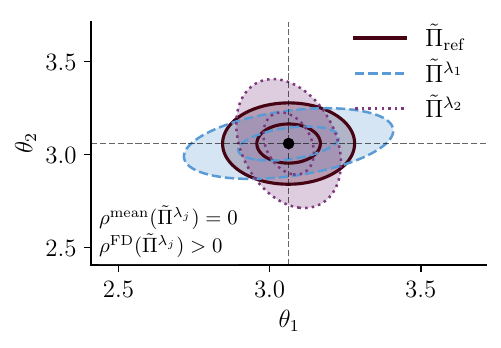}
        \caption{}
        \label{fig:comparison_same_mean_diff_cov}
    \end{subfigure}
    \begin{subfigure}[b]{0.32\linewidth}
        \centering
        \includegraphics[width=\linewidth]{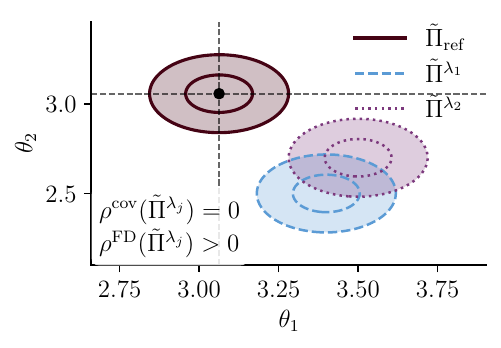}
        \caption{} \label{fig:comparison_same_cov_diff_mean}
    \end{subfigure}
    \hfill
    \begin{subfigure}[b]{0.32\textwidth}
        \centering
        \includegraphics[width=\linewidth]{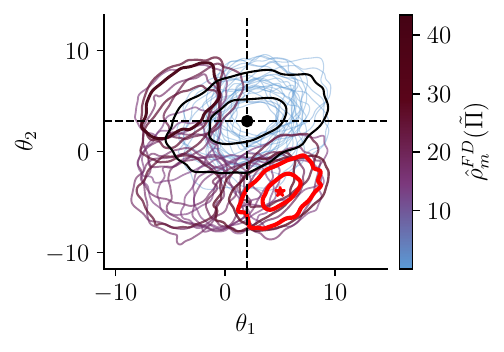}
        \caption{}
        \label{fig:toy_gaussian_model_multivariate}
    \end{subfigure}
    \caption{\emph{Strength of measure for the multivariate Gaussian model.} (a) Posteriors with the same mean, but different covariances, (b) Posteriors with different means, but the same covariances. (c) A set of candidate multivariate Gaussian distributions. \textcolor{black}{\rule[0.5ex]{1.2em}{1.0pt}} is the $\pi_\text{ref}$ and \textcolor{red}{\rule[0.5ex]{1.2em}{1.0pt}} is the optimal $\pi^\lambda$.
    }
    \label{fig:comparison_measure_strength}
\end{figure}
We now extend the example to bivariate Gaussian posteriors, with reference $\trueposterior=\mathcal{N}(\mu_{\mathrm{ref},n},\Sigma_{\mathrm{ref},n})$ and candidates $\tilde{\Pi}^{\lambda_j}=\mathcal{N}(\mu_{j,n},\Sigma_{j,n})$. \Cref{fig:comparison_measure_strength} illustrates how the FD detects distributional changes that sensitivity measures based on individual posterior expectations can miss.
When the candidates have the same mean as the reference but different covariances, $\rho^{\mathrm{mean}}(\widetilde\Pi^{\lambda_j})=0$ whereas $\rho^{\mathrm{FD}}(\tilde{\Pi}^{\lambda_j}) = 
\mathrm{tr}(
(\Sigma_{j,n}^{-1}-\Sigma_{\mathrm{ref},n}^{-1})^\top
(\Sigma_{j,n}^{-1}-\Sigma_{\mathrm{ref},n}^{-1})\,
\Sigma_{\mathrm{ref},n}
) > 0$ for $j\in \{1,2\}$
(see  \Cref{fig:comparison_same_mean_diff_cov}). 
Conversely, when only the means differ, covariance-based sensitivity vanishes whilst
$\rho^\FD(
\tilde{\Pi}^{\lambda_j})
=
\|
\Sigma_{j,n}^{-1} (\mu_{\text{ref},n}-\mu_{j,n} )\|_2^2>0$ (see \Cref{fig:comparison_same_cov_diff_mean}). Although this issue could be remedied through a sensitivity measure based on both moments in this example, most distributions are not characterised through a finite number of moments, and using a statistical divergence is therefore advantageous.
We show this by considering the worst-case sensitivity obtained through 
box constraints in $\Gamma
=
\{
[\lambda_0,\Lambda_1]:
\lambda_0=\Sigma^{-1}\mu,
\;
\Lambda_1=-\tfrac{1}{2}\Sigma^{-1},
\mu_{0},\mu_{1}\in[-4,5],
\sigma_{00},\sigma_{11}\in[2,4],
\sigma_{01}=\sigma_{10}\in[0,2]
\}$ in
\Cref{fig:toy_gaussian_model_multivariate}, where the worst-case prior identified by the FD consists of changes in both mean and covariance as expected.

\paragraph{Comparison with alternative sensitivity measures.}
We now move on to multivariate Gaussians and study the impact of $d_\Theta$ and $m$ on the estimation error and computational cost of different sensitivity methods based on the posterior mean, the Wasserstein-2 distance and the KL divergence. Although closed-form KL and Wasserstein distances are available for this Gaussian experiment, we intentionally compare sample-based estimators that mirror the estimators one would need in non-conjugate settings. Our results are given in \Cref{fig:comparison_measure_sample_complexity}.

We begin by comparing the FD and mean measures. These both have similar estimation errors of $\mathcal{O}(m^{-1/2})$, and relatively small computational cost of $O(m d_{\Theta}^2)$ and $\mathcal{O}(m d_\Theta)$ respectively. Although the FD will have a higher cost for large $d_\Theta$, this has to be balanced with the fact that it controls changes in the entire posterior. 

In contrast, the two other divergence-based measures are able to control changes in the entire posterior, but have estimation errors and computational costs which are orders of magnitude larger than that of the FD. Firstly, the KL divergence is representative of a broad class of divergences requiring access to normalised densities, which are unavailable for many posteriors. We therefore estimate it by applying kernel density estimation to samples from $\trueposterior$ and $\candidateposterior$, then plugging in these estimated densities to approximate the divergence through a Monte Carlo estimator. This costs $O(m^2 d_{\Theta})$ and deteriorates with dimension: in \Cref{fig:comparison_measure_sample_complexity}, convergence is not observed for $d_\Theta=25$ and $d_\Theta=100$ for $m < 5000$. Secondly, we consider the Wasserstein-2 distance, which is typically estimated by optimal transport between empirical samples from $\trueposterior$ and $\candidateposterior$. For $d_{\Theta}>1$, this costs $\mathcal{O}(m^3+m^2 d_{\Theta})$, while its convergence rate is $O(m^{-1/d_{\Theta}})$ for $d_{\Theta}>2$ \citep{Fournier2015}, illustrating its poor scaling with $m$ and $d_\Theta$.

\begin{figure}[t!]
    \centering
    \begin{subfigure}[b]{0.24\linewidth}
        \centering
        \includegraphics[width=\linewidth]{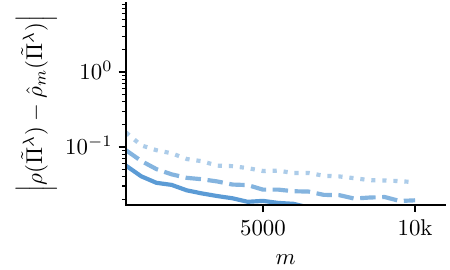}
    \end{subfigure}   
    \begin{subfigure}[b]{0.24\linewidth}
        \centering
        \includegraphics[width=\linewidth]{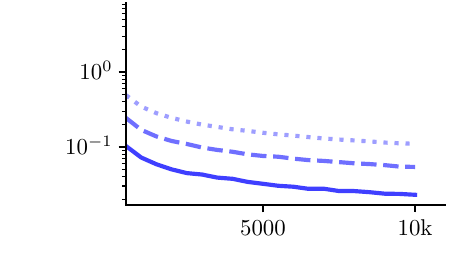}
    \end{subfigure}
    \begin{subfigure}[b]{0.24\linewidth}
        \centering
        \includegraphics[width=\linewidth]{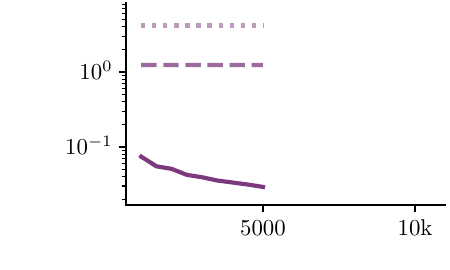}
    \end{subfigure}
        \begin{subfigure}[b]{0.24\linewidth}
        \centering
        \includegraphics[width=\linewidth]{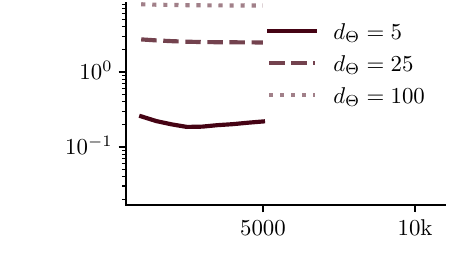}
    \end{subfigure}

    \vspace{0.4em}

    \setcounter{subfigure}{0}

        \begin{subfigure}[b]{0.24\linewidth}
        \centering
        \includegraphics[width=\linewidth]{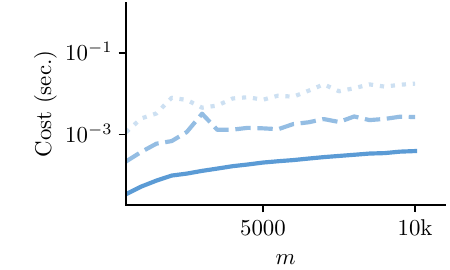}
        \caption{$\widehat{\rho}_m^{\text{FD}}(\candidateposterior)$}
    \end{subfigure}
        \begin{subfigure}[b]{0.24\linewidth}
        \centering
        \includegraphics[width=\linewidth]{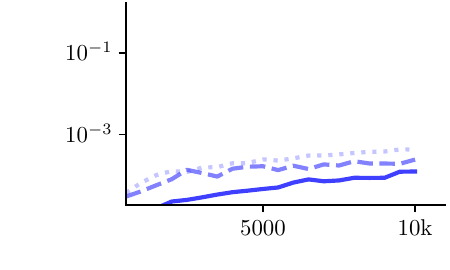}
        \caption{$\widehat{\rho}_m^{\text{mean}}(\candidateposterior)$}
    \end{subfigure}
        \begin{subfigure}[b]{0.24\linewidth}
        \centering
        \includegraphics[width=\linewidth]{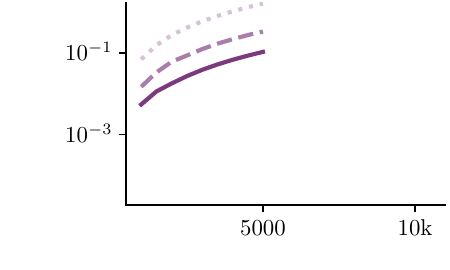}
        \caption{$\widehat{\rho}_m^{\text{KL}}(\candidateposterior)$}
    \end{subfigure}
    \begin{subfigure}[b]{0.24\linewidth}
        \centering
        \includegraphics[width=\linewidth]{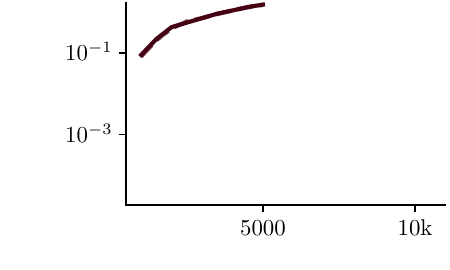}
        \caption{$\widehat{\rho}_m^{\text{W}_2}(\candidateposterior)$}
    \end{subfigure}
    \caption{\emph{Finite-sample complexity (top row) and computational cost (bottom row) across $\paramdim$ for the multivariate Gaussian model.} 
    Results are averaged over 500 runs. The experiments for the Wasserstein-2 distance and the KL divergence were too expensive to run up to $m=10000$, so were run for only $m=5000$ samples.}
    \label{fig:comparison_measure_sample_complexity}
\end{figure}

Before concluding, we note that these computational differences are amplified in global sensitivity analysis, where the measure must be evaluated repeatedly during optimisation. Mean-, KL-, and Wasserstein-based measures require samples from every candidate posterior, which explains why these are typically considered with discrete and finite neighbourhoods. In contrast, the FD reuses samples from $\trueposterior$ and requires only the candidate scores, making optimisation over continuous neighbourhoods considerably more tractable.

\subsection{Measuring Sensitivity to Learning Rate Estimation}\label{subsec:gen-bayes-ising}

We are now ready to move on to more realistic sensitivity problems. Choosing the learning rate, which controls the relative weight of the loss and prior, is a central challenge in generalised Bayesian inference, with no consensus among existing methods \citep{wu2023comparison}. It is also well established that generalised posteriors can be highly sensitive to the learning rate, and that some of these estimators can be somewhat unstable in low data regimes. Interestingly, our proposed FD sensitivity measure can be used to verify this sensitivity numerically. 

We illustrate this using the Ising model and dataset of \citet{matsubara2024generalized}. This is a challenging inference problem since the likelihood is only known in unnormalised form, which makes standard posteriors doubly intractable. The data comprise $n=1000$ binary configurations on a $6 \times 6$ grid graph, so that $\datadim = 36$, and we use a $\chi^2$ prior with $3$ degrees of freedom. We consider two generalised posteriors based on the pseudo-likelihood (PL; \cite{besag1974spatial}) and on the discrete Fisher divergence (DFD;  \cite{matsubara2024generalized}).
For each posterior, we used $m=5000$ samples from a reference posterior obtained via Hamiltonian Monte Carlo with $1000$ burn-in samples.

We compare three learning rate estimation methods by \citet{syring2019calibrating}, \citet{lyddon2019general} and \citet{matsubara2024generalized}, which all aim to ensure the right frequentist coverage for credible regions of the generalised posterior. Each produces a reference learning rate $\lambda_{\text{ref}}$ and corresponding $\trueposterior$, and we measure sensitivity over $\Gamma = \{\lambda >0 :|\lambda - \lambda_{\text{ref}}| \leq \epsilon \}$ for some $\epsilon > 0$. Measuring sensitivity is particularly important here because all three of these methods depend on algorithmic choices which can substantially affect the resulting estimate of the learning rate, such as a number of bootstrap replications or the number of steps of a numerical optimiser. The FD is particularly computationally convenient for this task; using decomposition 1 in \Cref{subsubsec:interpretability-through-factorisation}, we see that $
    \rho^{\text{FD}}(\tilde{\Pi}^{\lambda}) = (\lambda - \lambda_\mathrm{ref})^2 \mathbb{E}_{\theta \sim \tilde{\Pi}_{\text{ref}}}\![\| \nabla_{\theta} l(\theta;x_{1:n}) \|_2^2 ]$, $S^\FD(\Gamma) = \epsilon^2  \mathbb{E}_{\theta \sim \tilde{\Pi}_{\text{ref}}}\![\| \nabla_{\theta} l(\theta;x_{1:n})\|^2 ]$, and therefore no optimisation is needed.

\begin{figure}[t!]
    \centering
    \captionsetup[subfigure]{labelformat=empty}

    \begin{subfigure}{0.32\textwidth}
        \centering
        \includegraphics[width=\linewidth]{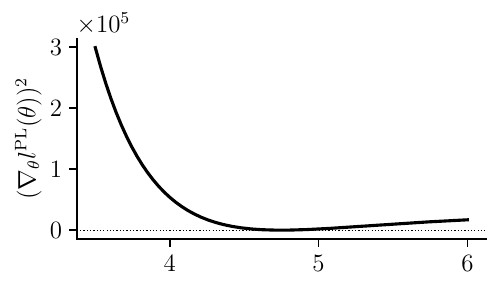}   
        \caption{}
        \label{fig:ising-loss-gradient-pseudolikelihood-1000}
    \end{subfigure}
    \hfill
    \begin{subfigure}{0.32\textwidth}
        \centering
        \includegraphics[width=\linewidth]{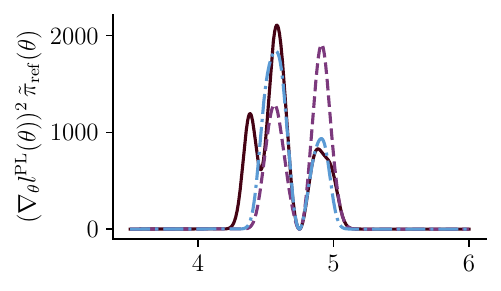}
        \caption{}
        \label{fig:ising-loss-gradient-times-density-pseudolikelihood-1000}
    \end{subfigure}
    \hfill
    \begin{subfigure}{0.32\textwidth}
        \centering
        \includegraphics[width=\linewidth]{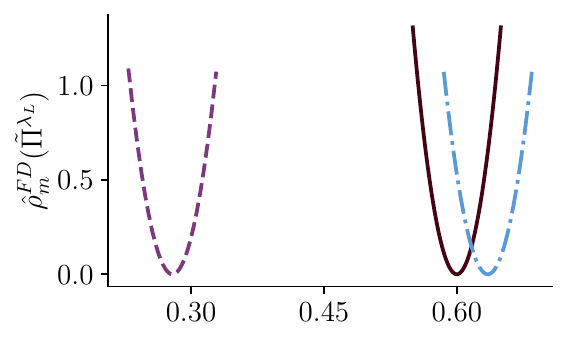}
        \caption{}
        \label{fig:ising-lr-comparison-pseudolikelihood-1000}
    \end{subfigure}
    
    \captionsetup[subfigure]{labelformat=empty}
    \vspace{-0.5cm}

    \begin{subfigure}{0.32\textwidth}
        \centering
        \includegraphics[width=\linewidth]{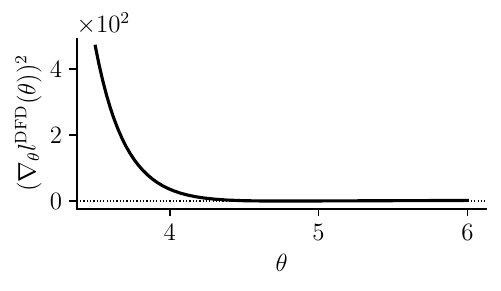}   
        \caption{}
        \label{fig:ising-loss-gradient-dfd-1000}
    \end{subfigure}
    \hfill
    \begin{subfigure}{0.32\textwidth}
        \centering
        \includegraphics[width=\linewidth]{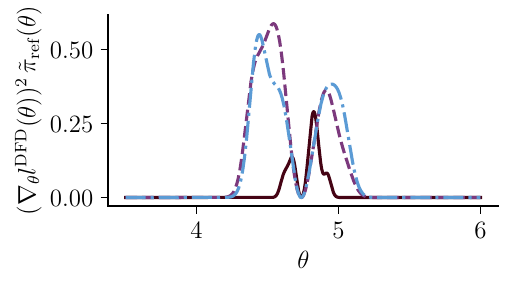}
        \caption{}
        \label{fig:ising-loss-gradient-times-density-dfd-1000}
    \end{subfigure}
    \hfill
    \begin{subfigure}{0.32\textwidth}
        \centering
        \includegraphics[width=\linewidth]{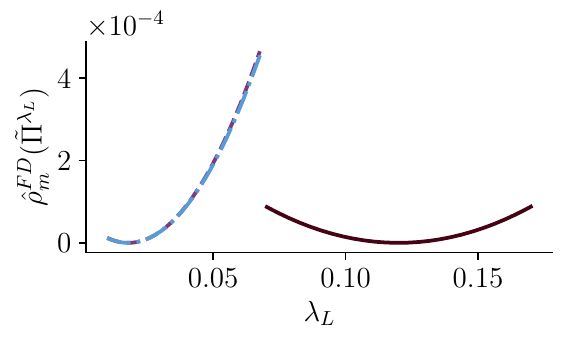}
        \caption{}
        \label{fig:ising-lr-comparison-dfd-1000}
    \end{subfigure}
    \vspace{-0.5cm}

    \caption{\emph{Generalised Bayesian inference for an Ising model}. Top: Sensitivity with PL loss. Bottom: Sensitivity with DFD loss. 
    \textcolor{darkcherry}{\rule[0.5ex]{1.2em}{1.0pt}} is \citet{matsubara2024generalized},
    \textcolor{purple}{\hdashrule[0.5ex]{1.2em}{1pt}{2pt 1pt}} is \citet{syring2019calibrating},
    and 
    \textcolor{blue}{\hdashrule[0.5ex]{1.2em}{1pt}{0.5pt 1.5pt}} is \citet{lyddon2019general}.
    First column: $(\nabla_\theta l^\mathrm{PL}(\theta))^2$.
    Second column: $(\nabla_\theta l^\mathrm{PL}(\theta))^2 \times \tilde{\pi}_\mathrm{ref}(\theta)$.
    Third column: $\hat{\rho}_m^{\FD}(\tilde{\Pi}^{\lambda_L})$ across the neighbourhood $\Gamma$ with $\epsilon=0.05$. 
    }
        
    \label{fig:ising-optimisation} 
\end{figure}

In \Cref{fig:ising-optimisation}, we present the components within the expectation, contributing to the sensitivity measure $\rho^{\FD}(\tilde{\Pi}^{\lambda})$. First, we note that in this setup the absolute sensitivity values for the PL and DFD losses are not directly comparable, since the magnitudes of the squared loss gradients differ. Therefore, we compare the calibration procedures separately for each loss and dataset, i.e. within each row.
For the PL loss, the calibration procedure of \citet{matsubara2024generalized} leads to the largest sensitivity. This occurs because the integrand in \Cref{fig:ising-loss-gradient-times-density-pseudolikelihood-1000} has three pronounced bumps, resulting in a substantially larger area under the curve than for the approaches of \citet{syring2019calibrating} and \citet{lyddon2019general}. The left-most bump appears because $\trueposterior$ contains a small left-tail bump, which is amplified by the rapidly increasing squared loss gradient.
For the DFD loss, the calibration procedures of \citet{syring2019calibrating} and \citet{lyddon2019general} once again behave similarly, whereas the approach of \citet{matsubara2024generalized} provides the least sensitive posterior.

Before concluding, we emphasise the computational efficiency with which such plots can be generated. Thanks to the simplified expression for the FD, both  $\hat{\rho}^\text{FD}_m(\tilde{\Pi}^{\lambda_L})$ and $\widehat S_m^{\FD}(\Gamma)$ can be estimated from a single MCMC estimator using samples from $\trueposterior$, which takes around $10^{-4}$ seconds. This cost is therefore negligible relative to the roughly $7$ seconds it takes to sample. 

\subsection{Autoregressive model of the temperature in Kilpisjärvi}\label{subsec:heavy-tailed-ar-model-kilpisjarvi}

We now consider prior sensitivity analysis in the context of time-series models  \citep{karlsson2013forecasting, giannone2015prior}.
The aim of this experiment is twofold: to illustrate the method on a higher-dimensional posterior with real data, and to assess the performance of our method under different optimisation routines. 
We consider an autoregressive (AR) model for a univariate real-world time series of June temperatures recorded in Kilpisjärvi (Finland) starting in $1952$, with $n=62$ observations; see \Cref{fig:ark_posterior_predictive}. Such temperature records are commonly analysed to study long-term climate dynamics, and Bayesian time-series models are particularly useful here because the amount of data is limited. Prior sensitivity is important since downstream conclusions are based on posterior predictive distributions, and sensitivity to plausible prior perturbations may substantially affect predictive reliability.

The likelihood for an AR model of order $S \in \mathbb{N}$ is conditionally Gaussian and given by $p_\theta(x_{S+1:n}|x_{1:S}) = \prod_{i=S}^n p_\theta(x_i|x_{i-1},\ldots,x_{i-S})$,  $p_\theta(x_i|x_{i-1},\ldots,x_{i-S}) = \mathcal{N}(x_i;\alpha + \sum_{s=1}^S \beta_s\, x_{i-s}, \sigma^2)$,
and the parameters are $\theta = [\alpha,\beta_1,\dots,\beta_S,\sigma] \in \mathbb R^{\paramdim-1} \times (0, \infty)$. 
Following the analysis of this dataset in the \texttt{posteriordb} \citep{magnusson2025posteriordb} package, we take $\trueprior$ encoding independence across parameters: $\pi_\text{ref}(\alpha)= \mathcal N(\alpha;0,\,5^2),
\pi_\text{ref}(\beta_{s}) =  \mathcal N(\beta_s;0,\,5^2)$ for $s\in\{1,\ldots,S\}$, and $\pi_\text{ref}(\sigma) = \mathrm{HalfCauchy}(\sigma;0,\,1)$ where $\sigma>0$ is a scale parameter. We fix $S=5$, so that $\paramdim = 7$. As these priors act on different scales, we work with $z_j = \Phi^{-1}(F_{\mathrm{ref},j}(\theta_j))$, where $F_{\mathrm{ref},j}$ is the CDF of $\Pi_{\mathrm{ref},j}$ and $\Phi$ is the standard Gaussian CDF, so that $\pi_{\mathrm{ref}}(z)=\prod_{j=1}^{\paramdim}\mathcal N(z_j;0,1)$. We consider candidates $\pi(z|\lambda)=\prod_{j=1}^{\paramdim}\mathcal N(z_j;\mu_j,s_j^2)$ with $\lambda:=\lambda_\pi=[\lambda_\alpha,\lambda_{\beta_1},\dots,\lambda_{\beta_5},\lambda_\sigma]$, so that $\natparamdim=14$, and a box-constrained neighbourhood with $\mu_j\in[-2.0,2.0]$ and $s_j\in[0.25,4.0]$ for all $j$. Due to the composite structure of the prior, the box-constrained neighbourhood admits the decomposition
$
\Gamma = \prod_{j=1}^{\paramdim} \Gamma_j
$,
and hence the sensitivity measure decomposes across dimensions: 
$\widehat S_m^{\FD}(\Gamma) :=
\sum_{j=1}^{\paramdim} \widehat S_m^{\FD}(\Gamma_j)$.
Since $\trueprior$ lies in every $\Gamma_j$, $\inf_{\lambda_j\in\Gamma_j}\hat\rho^{\FD}_{m}(\lambda_j)=0$ for all $j$.
To perform inference, we draw $m = 5000$ posterior samples $\{ \theta_i \}_{i=1}^m$ from $\trueposterior$ with $1000$ burn-in samples and $5$ chains, using the No-U-Turn-Sampler (NUTS) in Stan \citep{carpenter2017stan}, mapped to $z$.  The results are presented in \Cref{fig:ark_param_and_cost} and \Cref{fig:ark_posterior_predictive}. 

One of the main advantages of our approach is that we obtain not only an estimate of $ S_m^{\FD}(\Gamma)$, but also the worst-case prior at which the supremum is obtained. As seen in \Cref{fig:kilpisjarvi_param_three_panel_priors}, the parameter $\lambda^{\mathrm{sup}}_{m}$ leads to highly concentrated priors with mass pushed away from $\trueprior$. 

Due to the use of independent prior, we can also use the per-dimension decomposition to evaluate parameter-wise contributions to the overall sensitivity $\widehat S_m^{\FD}(\Gamma)=118$; see \Cref{fig:kilpisjarvi_param_component_sensitivity}. Dominant contribution comes from sensitivity to the prior on $\sigma$, which accounts for approximately $16\%$ of $\widehat S_m^{\FD}(\Gamma)$, followed by sensitivity to the prior on $\beta_1$ and $\beta_2$, contributing around $14.4\%$ and $14.0\%$, respectively. This is consistent with intuition: the most recent observations have the strongest influence, so higher sensitivity to the first and second lag coefficients is expected. 

The worst-case prior has little visible impact at the level of the posterior predictive distribution; see \Cref{fig:ark_posterior_predictive}. Under both $\trueprior$ and the worst-case prior, the predictive means follow almost identical trajectories, and the $95\%$ credible regions have comparable widths. This suggests that, although the FD identifies the prior perturbation with the largest effect on the posterior, the likelihood is informative enough that this perturbation does not propagate to predictions. For this model and neighbourhood, predictive conclusions are therefore robust to plausible changes in the prior.

\begin{figure}[t!]
  \centering
  \begin{subfigure}[t]{0.35\textwidth}
    \centering
    \includegraphics[width=\linewidth]{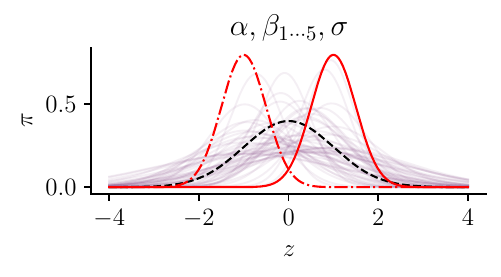}
    \caption{}
    \label{fig:kilpisjarvi_param_three_panel_priors}
  \end{subfigure}
  \begin{subfigure}[t]{0.25\textwidth}
    \centering
    \includegraphics[width=\linewidth]{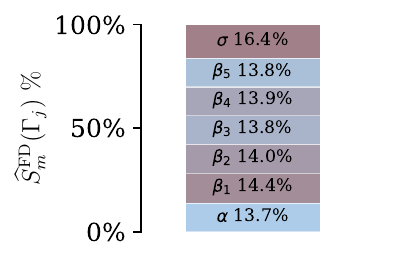}
    \caption{}
    \label{fig:kilpisjarvi_param_component_sensitivity}
  \end{subfigure}
  \begin{subfigure}[t]{0.35\textwidth}
    \centering
    \includegraphics[width=\linewidth]{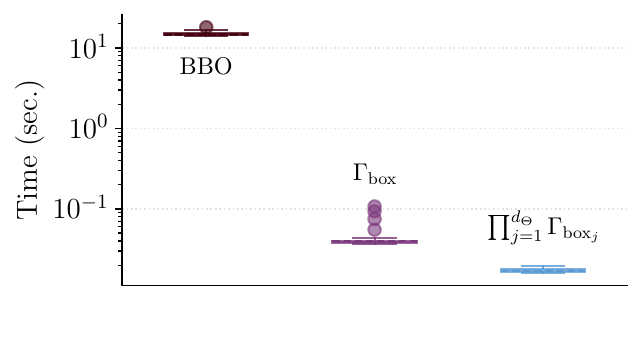}
    \caption{}
    \label{fig:kilpisjarvi_computational_cost}
  \end{subfigure}
  \caption{\emph{Kilpisjärvi AR(5) model optimisation results}. (a) \textcolor{black}{%
    \rule[0.5ex]{0.25em}{1.0pt}\hspace{0.15em}%
    \rule[0.5ex]{0.25em}{1.0pt}\hspace{0.15em}%
    \rule[0.5ex]{0.25em}{1.0pt}%
    } is $\pi_\text{ref}$
    and 
    \textcolor{red}{\rule[0.5ex]{1.2em}{1.0pt}},
    \textcolor{red}{%
    \rule[0.5ex]{0.25em}{1.0pt}\hspace{0.15em}%
    \rule[0.5ex]{0.25em}{1.0pt}\hspace{0.15em}%
    \rule[0.5ex]{0.25em}{1.0pt}%
    } 
are the optimal $\pi^{\lambda^{\mathrm{sup}}_{m}}$, found among all candidate $\pi^{\lambda_{m}}$ (faint \textcolor{purple}{\rule[0.5ex]{1.2em}{1.0pt}} lines).   (b) $\widehat S_m^{\FD}(\Gamma_j)$ contribution to $\widehat S_m^{\FD}(\Gamma)$.  (c) Optimisation runtime in log-scale (boxplot obtained from 500 repetitions) of our convex box-constrained (\textcolor{darkcherry}{\textbullet}), decomposed box-constrained (\textcolor{purple}{\textbullet}) and black-box (BBO) (\textcolor{darkblue}{\textbullet}) optimisation methods.}
  \label{fig:ark_param_and_cost}
\end{figure}

Before concluding this example, we also comment on the computational cost of running our procedure, with results presented in \Cref{fig:kilpisjarvi_computational_cost}.
We show three versions: a standard black-box dual annealing optimisation, a method which optimises by evaluating the quadratic-form objective in all $2^{\natparamdim} = 2^{14}$ corner points of the box constraints, and finally a method which uses the decomposition of the sensitivity measure as a sum of quadratic forms to consider only $\sum_{j=1}^{\paramdim} 2^{d_{\Lambda_j}} = 28$ corner configurations. 
The black-box dual annealing approach runs for substantially longer than the alternatives, and is primarily used to benchmark the computational advantages of using the quadratic-form structure of our objective. Its runtime is comparable to that of MCMC, which takes approximately $15$ seconds to run for this problem. In contrast, the most general quadratic form approach is around $500$ times cheaper to run, whilst the sum of quadratic forms is around  $1000$ times cheaper to run. This clearly demonstrates that the computational cost of performing sensitivity analysis is negligible relative to that of MCMC.

\begin{figure}[t!]
    \centering
    \begin{subfigure}[t]{0.45\textwidth}
        \centering
        \includegraphics[width=\linewidth]{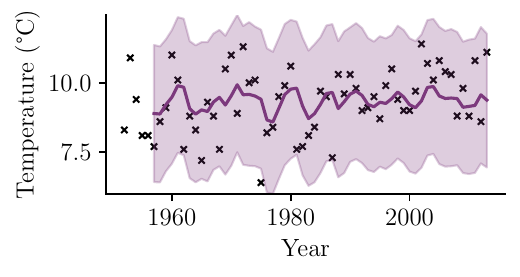}   
        \caption{Posterior predictive under $\trueprior$.}
        \label{fig:kilpisjarvi-posterior-predictive-ref}
    \end{subfigure}
    \begin{subfigure}[t]{0.43\textwidth}
        \centering
        \includegraphics[width=\linewidth]{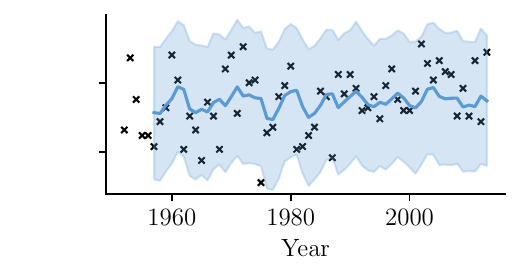}
        \caption{Posterior predictive under the worst-case prior.}
        \label{fig:kilpisjarvi-posterior-predictive-corner}
    \end{subfigure}
    \caption{\emph{Kilpisjärvi model.}
    \textbf{$\times$} are true observations. Mean and 95\% credible regions.}
    \label{fig:ark_posterior_predictive}
\end{figure}

\subsection{Simulation-based inference for radio propagation modelling}\label{subesc:sbi}
Our final example showcases our approach on a highly challenging Bayesian inference problem and demonstrates a possible advantage of reparametrisation. We consider simulation-based inference (SBI) for the Turin radio propagation model \citep{turin1972statistical}. The Turin model is an early stochastic radio channel model commonly used to help engineers evaluate communication systems across different propagation environments without collecting costly real-world measurements. Although it was first proposed in the 1970s, it remains widely used for modern millimetre-wave systems such as 5G \citep{haneda2015statistical, samimi20163}.

The Turing model is a simulator generating complex-valued time-series observations and depending on $\paramdim = 4$ parameters
$\theta = [G_0, T, \nu, \sigma_W^2]$. We consider the setting in \citet{huang2023learning}, where a single time-series of length $801$ and bandwith $4 \text{GHz}$ is observed. We perform SBI through neural likelihood estimation (NLE) with a masked-autoregressive flow and a pre-trained summary network reducing the observation to a six-dimensional statistic. We generate $m = 5000$ posterior samples from $\trueposterior$ using Hamiltonian Monte Carlo with $10$ parallel chains, discarding the first $1000$ iterations per chain as burn-in.

Although the majority of the literature \citep{bharti2019estimator,Bharti2021,huang2023learning} works with independent uniform priors due to the simplicity of specifying upper and lower bounds per parameter, this may not be a good choice as some combination of parameters might be unrealistic from a physical perspective. For instance, when communication takes place in a room \citep{pedersen2018modeling},  $\nu$ and $T$ get positively coupled: an increase in $\nu$ corresponds to more reverberation of the signal inside the room, which suggests slower decay of the signal and hence larger $T$. Conversely, $\nu$ and $G_0$ are negatively correlated as a large signal variance in the observations can be explained by either increasing $G_0$ or decreasing $\nu$.

We therefore perform sensitivity analysis to the choice of prior. Following \citet{huang2023learning}, we start with reference uniform priors
$G_0 \sim \text{Unif}(10^{-9}, 10^{-8})$,
$T \sim \text{Unif}(10^{-9}, 10^{-8})$,
$\nu \sim \text{Unif}(10^{7}, 5 \times 10^{9})$, and
$\sigma_W^2 \sim \text{Unif}(10^{-10}, 10^{-9})$. Given two parameters $\theta_i,\theta_j$ that we suspect of being correlated, we then consider candidate priors $
\pi(\theta\mid\lambda) =
\pi_\text{ref}(\theta)
c_{\lambda}(\theta_i, \theta_j)$. This candidate family preserves the marginal distributions while introducing dependence through a copula density with hyperparameter $\lambda$:
\begin{align*}
c_{\lambda}(\theta_i, \theta_j)
=
\frac{1}{\sqrt{1-\lambda^2}}
\exp\!\left(
\frac{
2\lambda z_i z_j
-
\lambda^2(z_i^2+z_j^2)
}{
2(1-\lambda^2)
}
\right),
\end{align*}
where $z_i = \Phi^{-1}\!\bigl(F_i(\theta_i)\bigr)$, and $F_i(\theta_i)=\frac{\theta_i-a_i}{b_i-a_i}$ is the CDF of the reference prior. We consider neighbourhoods of the form $\Gamma = \{\lambda:|\lambda - \lambda_{\text{ref}}| \le \epsilon \}$, where $\lambda_{\mathrm{ref}}=0$. 

We examine FD sensitivity when introducing dependence between $(G_0, \nu)$ or $(T, \nu)$ with $\epsilon = 0.2$ in \Cref{fig:sbi-copula-full-axis} for physically natural directions, i.e. positive correlation between $(T, \nu)$ and negative correlation between $(G_0, \nu)$. We see that the FD is not convex in $\lambda$ and looking at a smaller neighbourhood (e.g. $\epsilon=0.1$) would lead to a worst-case which is not on the boundary. This illustrates a significant limitation of the informal approach to sensitivity analysis \citep{zhu2011bayesian,kurtek2015bayesian,giacomini2021robustsetidentified,ghaderinezhad2022wasserstein}: evaluating sensitivity over a discrete set of candidate posteriors $\mathcal{P}_{\Gamma} := \{\candidateposterior \in \mathcal{P}(\Theta) : \lambda \in  \Gamma\}$, arising from $\Gamma = \{ \lambda_{1}, 
\ldots, \lambda_{K} \}$ could lead us to miss the worst-case. 
\Cref{fig:sbi_turin_gaussian_copula_fd_grid_marked} illustrates this pitfall: with 
$\Gamma = \{-0.02, 0.05, 0.10, 0.14\}$, the informal procedure selects the local maximum $\lambda = 0.05$, while the global supremum at $\lambda = 0.2$ is missed entirely. Thankfully, unlike for other divergences, the low computational cost of the FD makes solving this one-dimensional optimisation problem tractable. This takes around $0.696$ seconds, which is negligible relative to the approximately $180$ seconds needed to run MCMC.

Before concluding, we note that \Cref{assump:fd} is actually violated as $s_{\pi} \notin L^2(\tilde{\Pi}_{\text{ref}})$. As a result, our previous analysis was in fact not based on a valid statistical divergence (see \Cref{appsubesc:sbi} for details). To satisfy \Cref{assump:fd}, one approach consists of reparametrising the model with $z$ instead of $\theta$. This reparametrisation is not only desirable from a theoretical point of view, but \Cref{fig:sbi_turin_gaussian_copula_normalized} also shows that it leads to near-convexity over our neighbourhood, and the worst-case is attained at a boundary point. 

\begin{figure}[t!]
    \centering
    \begin{subfigure}{0.3\textwidth}
        \centering
        \includegraphics[width=\linewidth]{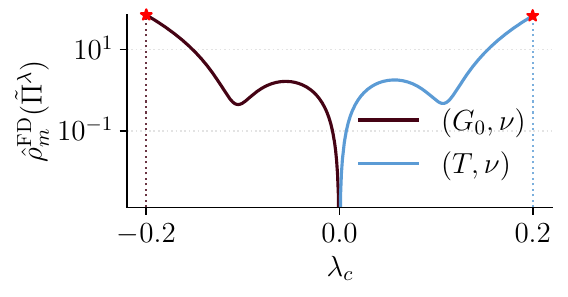}   
        \caption{}
        \label{fig:sbi-copula-full-axis}
    \end{subfigure}
    \begin{subfigure}{0.3\textwidth}
        \centering
        \includegraphics[width=\linewidth]{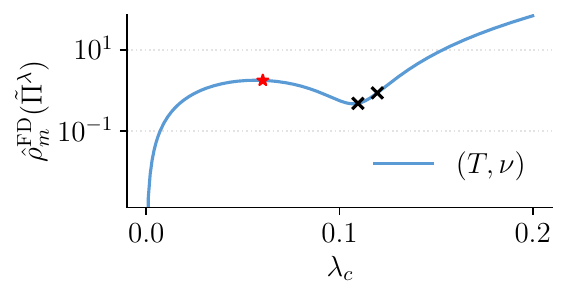}
        \caption{}
        \label{fig:sbi_turin_gaussian_copula_fd_grid_marked}
    \end{subfigure}
    \begin{subfigure}{0.3\textwidth}
    \centering
    \includegraphics[width=\linewidth]{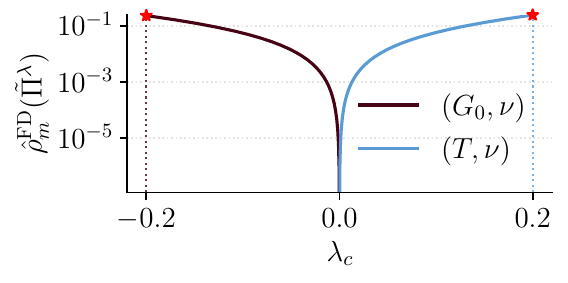}
    \caption{}
    \label{fig:sbi_turin_gaussian_copula_normalized}
    \end{subfigure}
    \caption{\emph{Prior sensitivity analysis for the Turin model with introduced prior dependence.} \textcolor{red}{$\star$} denotes the worst-case candidate identified via (a) FD-based optimisation, (b) informal sensitivity analysis and (c) FD-based optimisation with reparametrized $\theta$. $\times$ marks the evaluated candidates. The $y$-axis is on a log scale.}
    \label{fig:sbi-copula}
\end{figure}

\section{CONCLUSION}
This paper introduces a computationally tractable method for global Bayesian sensitivity analysis based on the FD. By requiring only a single set of posterior samples and operating through score functions, our method directly addresses the key computational limitations that have hindered the practical use of these types of methods. In many cases of practical relevance, we showed that the FD is a convex quadratic form, which allows the associated optimisation problems to be solved efficiently. In addition, because the FD considers discrepancies in the entire posterior, rather than merely summary statistics, it provides a strong measure of sensitivity which has high potential for challenging Bayesian inference problems. 

We believe there are many exciting extensions that could be considered. The method could be readily applied across a wide range of Bayesian inference problems, including for hierarchical models, as well as in multi-task \citep{guo2011sparse}, federated \citep{cao2023bayesian}, and transfer learning \citep{suder2025bayesian} settings. These settings are particularly relevant because model parameters are organised across levels (e.g., global and task-specific), and the FD objective can therefore be simplified under a per-level decomposition. This could enable sensitivity analysis of different levels of the model, allowing the practitioner to identify, for instance, whether posterior sensitivity is driven by shared hyperparameters or task-specific components, or whether heterogeneity across tasks amplifies sensitivity. 

Beyond this, several methodological extensions could be pursued in future work. A first direction would be to address settings involving multimodal posteriors with well-separated modes, a scenario for which the FD is known to perform poorly \citep{wenliang2020blindness, zhang2022towards}. One potential remedy would be to construct the divergence using a tempered version of the reference posterior. A second direction would be to extend our measure beyond hyperparameter sensitivity. For instance, \citet{berger1986robust, moreno1992bands, kurtek2015bayesian, ho2023global} study robustness under neighbourhoods defined by probability balls centred at a reference prior or likelihood; adapting our approach to these types of neighbourhoods could yield valuable insights. Finally, a third promising avenue would be to generalise the method beyond Euclidean parameter spaces to encompass posteriors defined over discrete domains, for which score functions are not naturally defined. In this case, the log-ratio matching divergence proposed by \citet{laplante2025conjugate} could offer a principled alternative and also leads to a convex quadratic form. Such an extension would be particularly relevant for Bayesian model choice and variable selection problems involving discrete indices, as well as for hidden Markov models with discrete latent states.

\subsubsection*{Acknowledgements}
The authors are thankful to Jeffrey Negrea and Chris Oates for helpful discussions, and to Ayush Bharti for support with the radio propagation example. AO was supported by a UCL EPSRC DTP Mathematical Science scholarship [EP/W524335/1], CD and FXB were supported by the EPSRC grant [EP/Y022300/1], and FXB was supported by [EP/Y011805/1].

\singlespacing
\bibliography{paper_ref}
\doublespacing
\newpage
\appendix

{
\begin{center}
\LARGE
    \textbf{Supplementary Materials}
\end{center}
}

The supplementary materials are as follows. In  \Cref{appsec:proofs}, we provide the proofs of all our theoretical results. In \Cref{appsec:div_derivation_quadratic_form}, we provide simplified assumptions and quadratic forms in various special cases. In \Cref{appsec:complementary_experiments}, we provide additional numerical experiments to complement the results in the main text.

\section{Proofs of Theoretical Results}\label{appsec:proofs}
\subsection{Proof of \Cref{prop:fd-div}} \label{appsubsec:fd-divergence}
\begin{proof}
We prove the result by verifying all conditions in Theorem 1 of \cite{zhang2022towards}: (a) $\tilde{\pi}_{\text{ref}}, \tilde{\pi} \in C^1(\Theta)$, (b) $\tilde{\pi}_{\text{ref}}$ and $ \tilde{\pi}$ have support on a common open connected subset of $\mathbb{R}^{d_\Theta}$ and (c) $s_{\tilde{\pi}_{\text{ref}}} - s_{\tilde{\pi}} \in L^2(\tilde{\Pi}_{\text{ref}})$.
\begin{itemize}
    \item \textbf{(a) Differentiability} Since $\tilde{\Pi} \in \mathcal{P}_{\FD}(\Theta)$, we have that $\tilde{\pi} \in C^1(\Theta)$. Moreover, by \Cref{assump:fd} (ii), $\pi_{\text{ref}}, L_{\text{ref}} \in C^1(\Theta)$. Since products of $C^1(\Theta)$ functions are also in $C^1(\Theta)$, it follows that $\tilde{\pi}_{\text{ref}} \in C^1(\Theta)$. 
\item \textbf{(b) Support} By definition of $\mathcal{P}_{\FD}(\Theta)$, $\operatorname{supp}(\tilde{\pi}) = \operatorname{supp}(\tilde{\pi}_{\text{ref}})$ and this is an open, connected subset of $\mathbb{R}^{d_\Theta}$ by \Cref{assump:fd} (i). 
\item \textbf{(c) Square integrability} By \Cref{assump:fd} (ii) we have that $s_{\pi_{\text{ref}}}, \nabla_\theta L_{\text{ref}}(\cdot; x_{1:n}) \in L^2(\tilde{\Pi}_{\text{ref}})$ and since $L^2(\tilde{\Pi}_{\text{ref}})$ is a vector space we obtain $
    s_{\tilde{\pi}_{\text{ref}}}(\theta)
= \nabla_\theta \log \tilde{\pi}_{\text{ref}}(\theta) 
= - \nabla_\theta L(\theta; x_{1:n}) + s_{\pi_{\text{ref}}}(\theta) \in L^2(\tilde{\Pi}_{\text{ref}})$.
Thus, since both $s_{\tilde{\pi}_{\text{ref}}}, s_{\tilde{\pi}} \in L^2(\tilde{\Pi}_{\text{ref}})$, we have that $s_{\tilde{\pi}_{\text{ref}}} - s_{\tilde{\pi}} \in L^2(\tilde{\Pi}_{\text{ref}})$.
\end{itemize}
\end{proof}

\subsection{Proof of \Cref{prop:mcmc-complexity}}\label{appsubsec:mcmc-complexity} 

The proof is based on Chebyshev's inequality and a bound on the variance of MCMC estimators with $V$-uniformly geometrically ergodic Markov chains due to \cite{durmus2024probability}. 
A stronger bound could be obtained by replacing Chebyshev's inequality with a bound on higher moments. However, this would require the existence of higher-order moments of the scores.

Before providing the proof, we present a preliminary lemma which guarantees convergence of second moments under MCMC sampling. To state this, we recall that
$L_{\sqrt V}
:=
\left\{
g:\Theta\to\mathbb{R} :
\|g\|_{\sqrt V}
:=
\sup_{\theta\in\Theta}
|g(\theta)|V(\theta)^{-1/2}
<\infty
\right\}$.

\begin{lemma}
\label{lem:vector-valued-mcmc}
Suppose \Cref{assumption:MCMC} holds. Let $r\in\mathbb{N}$ and let
$f=(f_1,\ldots,f_r)^\top:\Theta\to\mathbb{R}^r$
be measurable, with $f_j\in L_{\sqrt V}$ for every
$j\in\{1,\ldots,r\}$. Then there exists a constant
$ C(f;\xi)<\infty$, independent of $m$, such that $\mathbb{E}_{\xi}
[
\|
\frac{1}{m}\sum_{i=1}^m f(\theta_i)
-
\mathbb{E}_{\theta\sim\widetilde\Pi_{\mathrm{ref}}}
[f(\theta)]
\|_2^2
]
\leq
\frac{C(f;\xi)}{m}$.
\end{lemma}

\begin{proof}
Define the centred functions $f_j^{\mathrm{c}}(\theta)
:=
f_j(\theta)
-
\mathbb{E}_{\theta\sim\widetilde\Pi_{\mathrm{ref}}}
[f_j(\theta)]$ where $\mathbb{E}_{\theta\sim\widetilde\Pi_{\mathrm{ref}}}
[f_j^{\mathrm{c}}(\theta)]
=0$ for $j\in\{1,\ldots,r\}$.
We first verify that $f_j^{\mathrm{c}}\in L_{\sqrt V}$. 
Since $V(\theta)\geq e$ under \Cref{assumption:MCMC},
\begin{equation}
\begin{aligned}
\|f_j^{\mathrm{c}}\|_{\sqrt V}
&=
\sup_{\theta\in\Theta}
\frac{\left|
f_j(\theta)
-
\mathbb{E}_{\vartheta\sim\widetilde\Pi_{\mathrm{ref}}}
[f_j(\vartheta)]
\right|}{\sqrt{V(\theta)}}\leq
\|f_j\|_{\sqrt V}
+
e^{-1/2}
\left|
\mathbb{E}_{\vartheta\sim\widetilde\Pi_{\mathrm{ref}}}
[f_j(\vartheta)]
\right|.
\end{aligned}
\end{equation}
Furthermore,
\begin{equation}
\begin{aligned}
\left|
\mathbb{E}_{\vartheta\sim\widetilde\Pi_{\mathrm{ref}}}
[f_j(\vartheta)]
\right| \leq
\mathbb{E}_{\vartheta\sim\widetilde\Pi_{\mathrm{ref}}}
[|f_j(\vartheta)|]
\leq
\|f_j\|_{\sqrt V}
\mathbb{E}_{\vartheta\sim\widetilde\Pi_{\mathrm{ref}}}
[\sqrt{V(\vartheta)}]
\
&\leq
\|f_j\|_{\sqrt V}
\sqrt{\mathbb{E}_{\vartheta\sim\widetilde\Pi_{\mathrm{ref}}}
[V(\vartheta)]}
<\infty.
\end{aligned}
\end{equation}
The last inequality follows because the drift condition in
\Cref{assumption:MCMC} implies $\mathbb{E}_{\theta\sim\widetilde\Pi_{\mathrm{ref}}}
[V(\theta)]<\infty$.
Hence $f_j^{\mathrm c}\in L_{\sqrt V}$. For $m\in\mathbb{N}$, define
$S_{m,j}:=\sum_{i=1}^m f_j^{\mathrm c}(\theta_i)$.
By Theorem~2 of \citet{durmus2024probability}, specialised to $q=1$, and
Lemma~23 of \citet{durmus2024probability}, there exists a constant
$0<C_V<\infty$, depending only on the Markov chain and the Lyapunov
function $V$, such that
\begin{equation}
\mathbb{E}_{\xi}[S_{m,j}^2]
\leq
C_V
\left[
m
+
\mathbb{E}_{\theta\sim\xi}[V(\theta)]
+
\mathbb{E}_{\theta\sim\widetilde\Pi_{\mathrm{ref}}}[V(\theta)]
\right]
\|f_j^{\mathrm c}\|_{\sqrt V}^2.
\end{equation}
Dividing by $m^2$ gives
\begin{equation}
\begin{aligned}
\mathbb{E}_{\xi}
\left[
\left(
\frac{1}{m}S_{m,j}
\right)^2
\right]
&\leq
C_V
\left[
\frac{1}{m}
+
\frac{
\mathbb{E}_{\theta\sim\xi}[V(\theta)]
+
\mathbb{E}_{\theta\sim\widetilde\Pi_{\mathrm{ref}}}[V(\theta)]
}{m^2}
\right]
\|f_j^{\mathrm c}\|_{\sqrt V}^2
\\
&\leq
\frac{C_V}{m}
\left[
1
+
\mathbb{E}_{\theta\sim\xi}[V(\theta)]
+
\mathbb{E}_{\theta\sim\widetilde\Pi_{\mathrm{ref}}}[V(\theta)]
\right]
\|f_j^{\mathrm c}\|_{\sqrt V}^2,
\end{aligned}
\end{equation}
where the final inequality uses $m^{-2}\leq m^{-1}$ for
$m\geq1$. Therefore,
\begin{equation}
\mathbb{E}_{\xi}
\left[
\left(
\frac{1}{m}\sum_{i=1}^m f_j(\theta_i)
-
\mathbb{E}_{\theta\sim\widetilde\Pi_{\mathrm{ref}}}
[f_j(\theta)]
\right)^2
\right]
\leq
\frac{C(f_j;\xi)}{m},
\end{equation}
where $C(f_j;\xi)
:=
C_V
(
1
+
\mathbb{E}_{\theta\sim\xi}[V(\theta)]
+
\mathbb{E}_{\theta\sim\widetilde\Pi_{\mathrm{ref}}}[V(\theta)]
)
\|f_j^{\mathrm c}\|_{\sqrt V}^2$.
This constant is finite by \Cref{assumption:MCMC}, is independent of $m$, and
depends on $\xi$ only through
$\mathbb{E}_{\theta\sim\xi}[V(\theta)]$.
Finally,
\begin{equation}
\begin{aligned}
&\mathbb{E}_{\xi}
\left[
\left\|
\frac{1}{m}\sum_{i=1}^m f(\theta_i)
-
\mathbb{E}_{\theta\sim\widetilde\Pi_{\mathrm{ref}}}
[f(\theta)]
\right\|_2^2
\right]
\\
&\quad=
\sum_{j=1}^r
\mathbb{E}_{\xi}
\left[
\left(
\frac{1}{m}\sum_{i=1}^m f_j(\theta_i)
-
\mathbb{E}_{\theta\sim\widetilde\Pi_{\mathrm{ref}}}
[f_j(\theta)]
\right)^2
\right] \leq
\frac{1}{m}\sum_{j=1}^r C(f_j;\xi).
\end{aligned}
\end{equation}
The result follows upon setting
$
C(f;\xi):=\sum_{j=1}^r C(f_j;\xi)<\infty$.
\end{proof}
We are now ready to move on to the proof of \Cref{prop:mcmc-complexity}.
\begin{proof}
Define $f_\lambda(\theta)
:=
\left\|
s_{\widetilde\pi_{\mathrm{ref}}}(\theta)
-
s_{\widetilde\pi^\lambda}(\theta)
\right\|_2^2$. The growth condition is precisely the requirement that
$f_\lambda\in L_{\sqrt V}$. Applying
Lemma~\ref{lem:vector-valued-mcmc} with $r=1$ therefore gives a
constant $C(\widetilde\Pi^\lambda;\xi)<\infty$, independent of $m$,
such that
\begin{equation}
\mathbb{E}_{\xi}
\left[
\left|
\widehat{\operatorname{FD}}_m
\bigl(
\widetilde\Pi_{\mathrm{ref}}
\Vert
\widetilde\Pi^\lambda
\bigr)
-
\operatorname{FD}
\bigl(
\widetilde\Pi_{\mathrm{ref}}
\Vert
\widetilde\Pi^\lambda
\bigr)
\right|^2
\right] = \mathbb{E}_{\xi}
\left[
\left|
\frac{1}{m}\sum_{i=1}^m f_\lambda(\theta_i)
-
\mathbb{E}_{\theta\sim\widetilde\Pi_{\mathrm{ref}}}
[f_\lambda(\theta)]
\right|^2
\right]
\leq
\frac{C(\widetilde\Pi^\lambda;\xi)}{m},
\end{equation}
where $\frac{1}{m}\sum_{i=1}^m f_\lambda(\theta_i)
=
\widehat{\operatorname{FD}}_m
(
\widetilde\Pi_{\mathrm{ref}}
\Vert
\widetilde\Pi^\lambda
)$
and
$
\mathbb{E}_{\theta\sim\widetilde\Pi_{\mathrm{ref}}}
[f_\lambda(\theta)]
=
\operatorname{FD}
(
\widetilde\Pi_{\mathrm{ref}}
\Vert
\widetilde\Pi^\lambda
)$.
Hence, by Chebyshev's inequality,
for every $\tau>0$,
\begin{equation}
\begin{aligned}
&\mathbb{P}_{\xi}
\left(
\left|
\widehat{\operatorname{FD}}_m
\bigl(
\widetilde\Pi_{\mathrm{ref}}
\Vert
\widetilde\Pi^\lambda
\bigr)
-
\operatorname{FD}
\bigl(
\widetilde\Pi_{\mathrm{ref}}
\Vert
\widetilde\Pi^\lambda
\bigr)
\right|
\geq \tau
\right)
\\
&\quad\leq
\frac{1}{\tau^2}
\mathbb{E}_{\xi}
\left[
\left|
\widehat{\operatorname{FD}}_m
\bigl(
\widetilde\Pi_{\mathrm{ref}}
\Vert
\widetilde\Pi^\lambda
\bigr)
-
\operatorname{FD}
\bigl(
\widetilde\Pi_{\mathrm{ref}}
\Vert
\widetilde\Pi^\lambda
\bigr)
\right|^2
\right] \leq
\frac{C(\widetilde\Pi^\lambda;\xi)}
     {m\tau^2}.
\end{aligned}
\end{equation}
\end{proof}

\subsection{Proof of \Cref{prop:FD_convex_quadratic}}\label{app:proof_quadratic_form}
\begin{proof}
First, recall that score of $\candidateposterior$ is $s_{\tilde{\pi}^\lambda}(\theta)
=
- \lambda_L^\top \nabla_\theta l(\theta;x_{1:n},\lambda) + s_{\pi(\cdot|\lambda)}(\theta)$. In particular, defining $
J(\theta)
:=
[\,
-\nabla_\theta l(\theta_i;x_{1:n}),
\;
\nabla_\theta T(\theta)^\top
]$ and $\lambda
:= [
\lambda_L, 
\lambda_\pi]^\top$,
we can show that this score is linear in $\lambda$: $s_{\tilde{\pi}^\lambda}(\theta)
=
J(\theta)\,\lambda
+
s_g(\theta)$. Hence, the score difference between reference and candidate is also linear in $\lambda$, 
and:
\begin{align}
\hat{\rho}^{\FD}_m(\candidateposterior) = \estimatedFDposteriors
=
\frac{1}{m}
\sum_{i=1}^m
\Big\|
\big(
s_{\tilde{\pi}_{\mathrm{ref}}}(\theta_i)
-
s_g(\theta_i)
\big)
-
J(\theta_i)\lambda
\Big\|_2^2.
\label{appeq:fd_ref_only}
\end{align}
Expanding the square yields a quadratic form $
\hat{\rho}^{\FD}_m(\candidateposterior)
=
\lambda^\top A \lambda
+
b^\top \lambda
+
c$,
with $A$, $b$ and $c$ as defined in the statement of the result.

It now remains to check that this quadratic form is convex. For any vector $v \in \mathbb{R}^{d_\Lambda}$, 
\begin{align}
    v^\top A v = \frac{1}{m} \sum_{i=1}^m v^\top J(\theta_i)^\top J(\theta_i) v = \frac{1}{m} \sum_{i=1}^m \| J(\theta_i) v \|_2^2 \geq 0,
\end{align}
hence $A$ is positive semi-definite and quadratic form is convex, which completes the proof.
\end{proof}

\subsection{Proof of \Cref{prop:mcmc-complexity-uniform}}\label{appsubsec:mcmc-complexity-uniform} 

\begin{proof}
The proof proceeds in two steps. First, we show that the error in the estimated sensitivity is controlled by the estimation error of the quadratic coefficients. We then apply \Cref{lem:vector-valued-mcmc} to obtain $O(m^{-1})$ bounds for the second moments of the estimation error of these coefficients.

We start with step 1. Under \Cref{assump:minimal-exp-fam}, \Cref{prop:FD_convex_quadratic} implies that $
\widehat{\rho}^{\mathrm{FD}}_m(\widetilde\Pi^\lambda)
=
\lambda^\top A\lambda+b^\top\lambda+c$ and similarly $
\rho^{\mathrm{FD}}(\widetilde\Pi^\lambda)
=
\lambda^\top A^\star\lambda+b^{\star\top}\lambda+c^\star$,
where $A^\star=\mathbb{E}_{\theta \sim \trueposterior}[J(\theta)^\top J(\theta)]$, $b^\star =- 2 \mathbb{E}_{\theta \sim \trueposterior}[
        J(\theta)^\top
        (
        s_{\tilde{\pi}_{\mathrm{ref}}}(\theta)
        -
        s_g(\theta)
        )]$ and $c^\star = \mathbb{E}_{\theta \sim \trueposterior}[\|
s_{\tilde{\pi}_{\mathrm{ref}}}(\theta)
-
s_g(\theta)
\|_2^2]$. Define the coefficient estimatio errors
$
A_{\mathrm{diff}}
:=
A-A^\star
$
and
$
b_{\mathrm{diff}}
:=
b-b^\star.
$
Since adding a constant does not change the range of a function,
$c$ and $c^\star$ play no role in the sensitivity. Using the elementary inequality
$|\sup f-\sup g| \le \sup|f-g|$,
together with the analogous inequality for infima, we get
\begin{equation}
\left|
\widehat S^{\mathrm{FD}}_m(\Gamma)
-
S^{\mathrm{FD}}(\Gamma)
\right|
\leq
2
\sup_{\lambda\in\Gamma}
\left|
\lambda^\top A_{\mathrm{diff}}\lambda
+
b_{\mathrm{diff}}^\top\lambda
\right|.
\end{equation}
Since $\Gamma$ is compact, 
$
\exists R_\Gamma
=
\sup_{\lambda\in\Gamma}\|\lambda\|_2<\infty
$,
and therefore using $
\|A_{\mathrm{diff}}\|_{\mathrm{op}}
\leq
\|A_{\mathrm{diff}}\|_F
$ we get
\begin{equation}
\left|
\widehat S^{\mathrm{FD}}_m(\Gamma)
-
S^{\mathrm{FD}}(\Gamma)
\right|
\leq
2(R_\Gamma^2
\|A_{\mathrm{diff}}\|_{\mathrm{op}}
+
R_\Gamma
\|b_{\mathrm{diff}}\|_2 ) \leq
2(R_\Gamma^2
\|A_{\mathrm{diff}}\|_{F}
+
R_\Gamma
\|b_{\mathrm{diff}}\|_2).
\label{eq:sensitivity-bound}
\end{equation}
Applying Minkowski's inequality then gives
\begin{equation}
\left(\mathbb E_\xi
\left[
\left|
\widehat S^{\mathrm{FD}}_m(\Gamma)
-
S^{\mathrm{FD}}(\Gamma)
\right|^2
\right]\right)^{\frac{1}{2}}
 \leq 2 \left(R_{\Gamma}^2 \mathbb{E}_{\xi}\left[\|A_{\text{diff}}\|_{F}^2\right]^{\frac{1}{2}} + R_{\Gamma}\mathbb{E}_{\xi}\left[\| b_\text{diff}\|_2^2\right]^{\frac{1}{2}}\right).
\label{eq:sensitivity-mse}
\end{equation}
Hence it remains only to prove that $A_{\rm diff}$ and $b_{\rm diff}$ have $O(m^{-1})$ second moments.

First, define
$
f_A(\theta)
=
\operatorname{vec}(J(\theta)^\top J(\theta))
$, and we note that $\operatorname{vec}(A_{\rm diff}) =
\frac{1}{m}
\sum_{i=1}^m
f_A(\theta_i) - \mathbb{E}_{\xi}[f_A(\theta)]$.
The growth assumption in \Cref{prop:mcmc-complexity-uniform} implies
$
\|J(\theta)\|_F
\lesssim
V(\theta)^{1/4},
$
so
$
\|J(\theta)^\top J(\theta)\|_F \leq \|J(\theta)\|_F^2
\lesssim
V(\theta)^{1/2}.
$
Hence every coordinate of
$f_A$
belongs to
$L_{\sqrt V}$.
Applying \Cref{lem:vector-valued-mcmc} gives a constant
$C_A(\xi)<\infty$,
independent of $m$, such that
$
\mathbb E_\xi
\!\left[
\|A_{\mathrm{diff}}\|_F^2
\right]
=
\mathbb E_\xi
\!\left[
\|
\operatorname{vec}(A_{\mathrm{diff}})
\|_2^2
\right]
\leq
\frac{C_A(\xi)}{m}$.

Similarly, define
$
f_b(\theta)
=
-2
J(\theta)^\top
(s_{\widetilde\pi_{\mathrm{ref}}}(\theta)-s_g(\theta)).
$
Again the growth assumption implies
$
\|f_b(\theta)\|_2 \leq 2 \|J(\theta)\|_F ( \|s_{\widetilde\pi_{\mathrm{ref}}}(\theta)\|_2+\|s_g(\theta)\|_2)
\lesssim
V(\theta)^{1/2},
$
so every coordinate belongs to
$L_{\sqrt V}$.
A second application of \Cref{lem:vector-valued-mcmc} therefore yields a constant
$C_b(\xi)<\infty$,
independent of $m$, satisfying
$
\mathbb E_\xi[\|b_{\mathrm{diff}}\|_2^2]
\leq
\frac{C_b(\xi)}{m}.
$

Combining these bounds gives
$
\mathbb E_\xi
[
|
\widehat S^{\mathrm{FD}}_m(\Gamma)
-
S^{\mathrm{FD}}(\Gamma)
|^2
] \leq
\frac{C_{\Gamma,\xi}^2}{m}$ 
for some constant
$C_{\Gamma,\xi} = 2(R^2_{\Gamma} \sqrt{C_A(\xi)}+R_{\Gamma}\sqrt{C_b(\xi)})<\infty$.
The result now follows from Chebyshev's inequality, which gives $ \mathbb{P}_{\xi}(|\widehat S^{\mathrm{FD}}_m(\Gamma)-S^{\mathrm{FD}}(\Gamma)|>\frac{C_{\Gamma,\xi}}{\sqrt{m\delta}})\leq \delta$, and is equivalent to the two-sided bound in the statement of the theorem.
\end{proof}

\subsection{Proof of \Cref{prop:dominate_otherdivergences}}\label{app:proof_dominate_divergences}

\begin{proof}
   Fix any $\lambda\in\Gamma$. We will first consider the Wasserstein distance with $p=2$, which for 
$\trueposterior$ and $\candidateposterior$ is given by
$W_2(\trueposterior,\candidateposterior)
    :=
    \inf_{\mu \in \mathcal{H}(\trueposterior,\candidateposterior)}
    \left(
        \mathbb{E}_{(\theta,\theta')\sim \mu}\big[\|\theta-\theta'\|_2^2\big]
    \right)^{1/2}$, where $\mathcal{H}(\trueposterior,\candidateposterior)$ 
denotes the set of all couplings between $\trueposterior$ and $\candidateposterior$.
By Theorem 5.3 of \citet{huggins2018practical}, 
under part 1 of \Cref{assump:candidate-prior-regularity} there exists a constant
$\alpha_\lambda=\alpha(K_\lambda,R_\lambda,\tilde{\pi}^\lambda)>0$ such that
\begin{align} \label{eq:wass-fd-bound}
W_2(\trueposterior,\candidateposterior)
    \le 
    \frac{1}{\alpha_\lambda}\,\sqrt{\FD(\trueposterior\|\candidateposterior)}.
\end{align}
To complete the proof, we need to check that this bounds hold uniformly over the neighbourhood $\Gamma$.
Before taking supremums, we need to verify that the constant $\alpha_\lambda$ is uniformly bounded away from zero. By Remark 3.3 and proof of Proposition 3.4 of \cite{bolley2012convergence}, $\alpha_\lambda$ (C in their notation) depends on $\tilde{\pi}^\lambda$ only via $U^\lambda := \sup_{\|\theta\|\leq 3R_\lambda} -\log \tilde{\pi}^\lambda(\theta)$ and $L^\lambda := \inf_{\|\theta\|\leq 3R_\lambda} -\log \tilde{\pi}^\lambda(\theta)$ through coefficients of the form $\exp\left( \pm U^\lambda\right)$ and $\exp\left( -L^\lambda\right)$.
Therefore, as long as these are uniformly bounded above and below in $\Gamma$ we can ensure $\alpha_\lambda$ is uniformly bounded away from zero.
Recall that by part 2 of \Cref{assump:candidate-prior-regularity} we have that 
\begin{align}
0 < \tilde{\pi}^{R_\Gamma}_{\text{min}} := \inf_{\lambda \in \Gamma} \inf_{\|\theta\|_2\leq 3R_\lambda} \tilde{\pi}^\lambda(\theta) \leq \tilde{\pi}_{\text{max}}^{R_\Gamma} := \sup_{\lambda \in \Gamma} \sup_{\|\theta\|_2 \leq 3R_\lambda} \tilde{\pi}^\lambda(\theta) < \infty. 
\end{align}
Therefore, for every $\lambda \in \Gamma$ and $\| \theta\|_2 \leq 3R_\Gamma, -\log \tilde{\pi}^{R_\Gamma}_{\text{max}} \leq - \log \tilde{\pi}^\lambda(\theta)\leq - \log \tilde{\pi}^{R_\Gamma}_{\text{min}}$ and in particular for all $\lambda \in \Gamma$, $- \log \tilde{\pi}^{R_\Gamma}_{\text{max}} \leq L^\lambda \leq U^\lambda \leq - \log \tilde{\pi}^{R_\Gamma}_{\text{min}}$. Therefore, 
\begin{align}
    0 &< \frac{1}{\tilde{\pi}^{R_\Gamma}_{\text{max}}} \leq \inf_{\lambda \in \Gamma} \exp(L^\lambda) \leq \sup_{\lambda \in \Gamma}\exp(U^\lambda) \leq \frac{1}{\tilde{\pi}^{R_\Gamma}_{\text{min}}}  < \infty, \\
0 &< \tilde{\pi}^{R_\Gamma}_{\text{min}} \leq \inf_{\lambda} \exp(- U^\lambda)  \leq \sup_{\lambda \in \Gamma} \exp(-L^\lambda) \leq \tilde{\pi}^{R_\Gamma}_{\text{max}} < \infty.
\end{align}
Therefore, $\alpha_{\lambda}$ is uniformly 
bounded away from zero over $\Gamma$, i.e. $\alpha_{\Gamma} := \inf_{\lambda \in \Gamma} \alpha_\lambda > 0$. 
Taking supremum over $\lambda \in \Gamma$ in \Cref{eq:wass-fd-bound}, we obtain $
   S^{\text{W}_2}(\Gamma) \leq \frac{1}{\alpha_{\Gamma}} \sqrt{S^{\text{FD}}(\Gamma)}$.
The result for $p=1$ follows directly from the fact that $W_1(\Pi_1, \Pi_2) \leq W_2(\Pi_1, \Pi_2)$ for all $\Pi_1, \Pi_2 \in \mathcal{P}(\Theta)$ with finite second moments. As a result, we have $
    S^{\text{W}_1}(\Gamma) \leq S^{\text{W}_2}(\Gamma) \leq \frac{1}{\alpha_{\Gamma}} \sqrt{S^{\text{FD}}(\Gamma)}$.
    
We now move on to the total variation result. By Proposition 5.10 of \cite{huggins2018practical} and \Cref{eq:wass-fd-bound}, it follows that for all $\lambda \in \Gamma$,
    $\text{TV}(\trueposterior,\candidateposterior) \le \frac{1}{\sqrt{2\alpha_{\lambda}}} \sqrt{\FD(\trueposterior\|\candidateposterior)}$
and hence,  $
    S^{\text{TV}}(\Gamma) \leq \frac{1}{\sqrt{2\alpha_{\Gamma}}}  \sqrt{S^{\text{FD}}(\Gamma)}$.
    
Finally, we consider the KL sensitivity measure. It follows from \cite{koehler2022statistical} Proposition 1 and Theorem 1 that if $\tilde{\pi}^\lambda$ is $\alpha^\prime_\lambda$-strongly log-concave for all $\lambda \in \Gamma$ then it satisfies the log-Sobolev inequality with constant $C_{\text{LS}} \leq 1/(2 \alpha^\prime_\lambda)$ and $
    \text{KL}(\trueposterior \| \candidateposterior) \leq \frac{1}{2 \alpha^\prime_{\lambda}} \FD(\tilde{\Pi}_{\text{ref}}\| \tilde{\Pi}^\lambda)$. 
Taking supremum over $\lambda \in \Gamma$ and using the assumption that $\alpha^\prime_{\Gamma} := \inf_{\lambda \in \Gamma} K^\prime_{\lambda} > 0,$ we obtain 
$S^{\text{KL}}(\Gamma) \leq \frac{1}{2 \alpha^\prime_{\Gamma}} S^{\text{FD}}(\Gamma)$.
\end{proof}

\subsection{Proof of \Cref{prop:moment-bounds}}\label{appsubsec:interpretability-through-moment-control}
\begin{proof}
From Theorem 4.1 of \citet{huggins2018practical}, it follows that for any $\varepsilon \ge 0$ , if $W_2(\trueposterior,\candidateposterior)\le \varepsilon$, then 
$\|\mu_{\trueposterior}-\mu_{\candidateposterior}\|_2
    \le \varepsilon$ and $
    \|\Sigma_{\trueposterior}-\Sigma_{\candidateposterior}\|_2
    \le
3\,\min\!\big(\|\Sigma_{\trueposterior}\|_2^{1/2},\,\|\Sigma_{\candidateposterior}\|_2^{1/2}\big)\,\varepsilon
    + 5.25\,\varepsilon^2$, 
where $\|\cdot\|_2$ is the operator norm. 
Selecting $\varepsilon := \frac{1}{\alpha_{\lambda}}\,
\FD(\trueposterior\|\candidateposterior)^{1/2}$, it follows by \Cref{eq:wass-fd-bound} that 
\begin{align}
    \|\mu_{\trueposterior}-\mu_{\candidateposterior}\|_2
    &\le \frac{1}{\alpha_{\lambda}}\,
    \FD(\trueposterior\|\candidateposterior)^{1/2},\\
    \|\Sigma_{\trueposterior}-\Sigma_{\candidateposterior}\|_2
    &\le
    \frac{3}{\alpha_{\lambda}}\,
    \min\!\big(\|\Sigma_{\trueposterior}\|_2^{1/2},\,\|\Sigma_{\candidateposterior}\|_2^{1/2}\big)\
    \FD(\trueposterior\|\candidateposterior)^{1/2}
    + \frac{5.25}{\alpha_{\lambda}^2}\,
\FD(\trueposterior\|\candidateposterior).
\end{align}

Taking the supremum over $\lambda\in\Gamma$ and using \Cref{prop:dominate_otherdivergences},
we obtain $
S^{\mathrm{mean}}(\Gamma)
\le
\frac{1}{\alpha_{\Gamma}}
\sqrt{S^{\FD}(\Gamma)}$.
Similarly, we have $
S^{\mathrm{Cov}}(\Gamma)
\le
\frac{3}{\alpha_{\Gamma}}\,
\sup_{\lambda\in\Gamma}
\min(\|\Sigma_{\trueposterior}\|_2^{1/2},\,\|\Sigma_{\candidateposterior}\|_2^{1/2})\,
\sqrt{S^{\FD}(\Gamma)}
+\frac{5.25}{\alpha_{\Gamma}^2}\,
S^{\FD}(\Gamma)$.
Finally, since for any $\lambda \in \Gamma$ we have $\min(\|\Sigma_{\trueposterior}\|_2^{1/2},\,\|\Sigma_{\candidateposterior}\|_2^{1/2}) \le \sup_{\lambda\in\Gamma}
\min(\|\Sigma_{\trueposterior}\|_2^{1/2},\,\|\Sigma_{\candidateposterior}\|_2^{1/2})$ and $x \mapsto \min(a,x)$ is monotone increasing, taking supremum over $\lambda \in \Gamma$ yields: 
\begin{align}
    \sup_{\lambda \in \Gamma} \min\big(\|\Sigma_{\trueposterior}\|_2^{1/2},\,\|\Sigma_{\candidateposterior}\|_2^{1/2}\big) \leq \min\left(
\sqrt{\|\Sigma_{\trueposterior}\|_2},
\ \sup_{\lambda\in\Gamma}\sqrt{\|\Sigma_{\candidateposterior}\|_2}
\right).
\end{align}
Therefore,
\begin{align}
S^{\mathrm{Cov}}(\Gamma)
\le
\frac{3}{\alpha_{\Gamma}}\,
\min\left(
\sqrt{\|\Sigma_{\trueposterior}\|_2},
\ \sup_{\lambda\in\Gamma}\sqrt{\|\Sigma_{\candidateposterior}\|_2}
\right)
\sqrt{S^{\FD}(\Gamma)}
+\frac{5.25}{\alpha_{\Gamma}^2}\,
S^{\FD}(\Gamma).
\end{align}
\end{proof}

\section{Special Cases: Sensitivity to either Prior or Loss}\label{appsec:div_derivation_quadratic_form}

We briefly discuss how our method simplifies if we only consider one of prior sensitivity or loss sensitivity, but not joint sensitivity. In each case, we discuss how \Cref{assump:fd} can be simplified, and derive simplified quadratic form representations for the FD.

\paragraph{Prior sensitivity}
When measuring only prior sensitivity, we have already noted that $S^{\FD}(\Gamma)$ simplifies. Here, the scores of the loss function have been cancelled out, and we only look at differences in the prior scores. For this reason, \Cref{assump:fd} can be simplified to requiring only that  $\pi,\pi_{\text{ref}} \in C^1(\Theta)$ and $s_{\pi},s_{\pi_{\text{ref}}} \in L^2(\tilde{\Pi}_{\text{ref}})$, but we do not require $L, L_{\text{ref}}  \in C^1(\Theta)$ and $\nabla L,\nabla L_{\text{ref}} \in L^2(\tilde{\Pi}_{\text{ref}})$ in this setting. 

The expression for the quadratic form also simplifies. When
$\lambda_L = \lambda_{L_{\mathrm{ref}}}$, the estimated FD becomes a function of the prior hyperparameter
$\lambda_\pi$ only. Thus, $\natparamdim=\natparamdimprior$. Substituting $\lambda_{L_{\mathrm{ref}}}$ into the quadratic form yields $\estimatedFDposteriors
=
\lambda_\pi^\top A_{\pi} \lambda_\pi
+
b_{\pi}^\top \lambda_\pi
+
c_{\pi}$
where the coefficients are given by
$A_{\pi}
:=
\frac{1}{m}
\sum_{i=1}^m
\nabla_\theta T(\theta_i)^\top
\nabla_\theta T(\theta_i)$,
$b_{\pi}
:=
-\frac{2}{m}
\sum_{i=1}^m
\nabla_\theta T(\theta_i)^\top
(
s_{\pi_{\mathrm{ref}}}(\theta_i)
-
s_g(\theta_i)
)$, and
$c_{\pi}
:=
\frac{1}{m}
\sum_{i=1}^m
\|
s_{\pi_{\mathrm{ref}}}(\theta_i)
-
s_g(\theta_i)
\|_2^2$.

\paragraph{Loss sensitivity}

Similarly, $S^{\FD}(\Gamma)$ simplifies when we only consider learning-rate sensitivity; see decomposition 1 in \Cref{subsubsec:interpretability-through-factorisation}. Here, the scores of the priors have been cancelled out, and we only look at differences in the scores of the losses. For this reason, \Cref{assump:fd} can be simplified to requiring only that   $L, L_{\text{ref}}  \in C^1(\Theta)$ and $\nabla L,\nabla L_{\text{ref}} \in L^2(\tilde{\Pi}_{\text{ref}})$ but we do not require $\pi(\cdot|\lambda),\pi_{\text{ref}} \in C^1(\Theta)$ and $s_{\pi(\cdot|\lambda)},s_{\pi_{\text{ref}}} \in L^2(\tilde{\Pi}_{\text{ref}})$ in this setting.

The expression for the quadratic form in \Cref{prop:FD_convex_quadratic} also simplifies. If 
$\lambda_\pi = \lambda_{\pi_{\mathrm{ref}}}$,
the estimated FD reduces to a univariate quadratic form in the learning
rate $\lambda_L$, given by $\estimatedFDposteriors
=
\lambda_L^\top A_{L} \lambda_L
+
b_{L}^\top \lambda_L
+
c_{L}$. The coefficients are
$A_L
:=
\frac{1}{m}
\sum_{i=1}^m
\nabla_\theta l(\theta_i;x_{1:n})
\,
\nabla_\theta l(\theta_i;x_{1:n})^\top$,
$b_L
:=
-2 A_L \lambda_{L_{\mathrm{ref}}}$, and
$c_L :=
\lambda_{L_{\mathrm{ref}}}^\top A_L \lambda_{L_{\mathrm{ref}}}
$.

\section{Additional Details on the Experiments}\label{appsec:complementary_experiments}

This section contains additional results and details of the experimental setups for \Cref{sec:experiments}. For each experiment, we also check that our assumptions are satisfied.

\subsection{Conjugate Gaussian location models}\label{appsubsec:toy_gaussian_model}

\paragraph{\Cref{assump:fd}.}
Firstly, $\trueposterior \equiv \mathcal{N}(\tilde{\mu}_{\text{ref}},\tilde{\Sigma}_{\text{ref}})$ and $
\mathrm{supp}(\tilde{\pi}_{\mathrm{ref}})
 =
\mathbb{R}^{\paramdim}$ is open and connected. 

Secondly, we can verify the integrability and differentiability assumptions using \Cref{appsec:div_derivation_quadratic_form}. We have $\pi_\mathrm{ref} \in C^1(\Theta)$ since the Gaussian density is infinitely differentiable.
For showing $s_{\pi_{\text{ref}}} \in L^2(\tilde{\Pi}_{\mathrm{ref}})$, recall that $s_{\pi_\mathrm{ref}}(\theta)
=-\Sigma_\mathrm{ref}^{-1}(\theta-\mu_\mathrm{ref})$ which is linear in $\theta$, therefore
$\|s_{\pi_{\text{ref}}}\|_2^2$ is quadratic in $\theta$. Since $\trueposterior$ is Gaussian, it has finite second moments, and hence $\|s_{\pi_{\text{ref}}}\|_{L^2(\tilde{\Pi}_{\mathrm{ref}})} < \infty$.

Thirdly, verify the conditions on  $\mathcal{P}_{\Gamma}$, and the argument is near identical.  Since $\candidateprior$ is multivariate Gaussian and the likelihood is Gaussian, both $\pi(\cdot|\lambda)$ and $L(\cdot;x_{1:n},\lambda)$ are $C^1(\Theta)$ and the support is open and connected. Additionally, the  $\|s_{\pi(\cdot|\lambda)}\|_2^2$ and $\|\nabla_\theta L(\cdot;x_{1:n},\lambda)\|^2_2$ are both quadratics in $\theta$, and therefore integrable under a multivariate Gaussian posterior $\tilde{\Pi}_{\mathrm{ref}}$.

\paragraph{\Cref{assump:minimal-exp-fam}.}
We only consider prior sensitivity and therefore verify the conditions on the prior. The candidate prior $\candidateprior$ is multivariate Gaussian and  can be written in natural exponential family form as
shown in the main text.

\paragraph{\Cref{assump:candidate-prior-regularity}.}
Firstly, for every $\lambda \in \Gamma$ we have $\tilde{\Pi}^\lambda=\mathcal{N}(\mu_n,\Sigma_n)$ and hence the negative log density is $-\log \tilde{\pi}^\lambda(\theta) =
\frac{1}{2}(\theta-\mu_n)^\top \Sigma_n^{-1}(\theta-\mu_n)
+C$ for some constant $C \in \mathbb{R}$. The Hessian is therefore given by $- \nabla^2 \log \tilde{\pi}^\lambda(\theta)=\Sigma_n^{-1}$, which does not depend on $\theta$. Consequently, for all $\theta \in \text{supp}(\tilde{\pi}_{\text{ref}})$ (i.e., $R_\lambda=R_\Lambda=0$ for all $\lambda \in \Lambda$), $- \nabla^2 \log \tilde{\pi}^\lambda(\theta)
\succeq
K_\lambda I_{\paramdim}$ with $K_\lambda$ being the smallest eigenvalue of $\Sigma_n^{-1}$. As a result $K_{\Gamma} := \inf_{\lambda \in \Gamma} K_\lambda >0$.  Secondly, since $R_\lambda=0$ for all $\lambda \in \Lambda$, the second part of the assumption reduces to checking that $0 < \inf_{\lambda \in \Gamma}\tilde{\pi}^\lambda({0}) \leq \sup_{\lambda \in \Gamma}\tilde{\pi}^\lambda({0}) < \infty$. As $\lambda \rightarrow \tilde{\pi}^\lambda(0)$ is continuous and $\Gamma$ is compact, both bounds hold.

\paragraph{Closed-form divergences}\label{appsubsubsec:toy_gaussian_model_closed_forms}
We recall that for 
$\tilde{\Pi}^{\lambda_1}=\mathcal{N}(\mu_{1,n},\Sigma_{1,n})$ and
$\tilde{\Pi}^{\lambda_2}=\mathcal{N}(\mu_{2,n},\Sigma_{2,n})$,
the KL divergence and Wasserstein-2 distances admit closed forms:
\begin{align}
\mathrm{KL}\!\left(
\tilde{\Pi}^{\lambda_1}
|| \; 
\tilde{\Pi}^{\lambda_2}
\right)
& =
\frac12
\Bigg[
\mathrm{tr}\!\big(\Sigma_{2,n}^{-1}\Sigma_{1,n}\big)
+
(\mu_{2,n}-\mu_{1,n})^\top
\Sigma_{2,n}^{-1}
(\mu_{2,n}-\mu_{1,n})
-
\paramdim
+
\log \left(\frac{\det \Sigma_{2,n}}{\det \Sigma_{1,n}}\right)
\Bigg],\\
W_2^2\!\left(
\tilde{\Pi}^{\lambda_1}
,
\tilde{\Pi}^{\lambda_2}
\right)
& =
\|\mu_{1,n}-\mu_{2,n}\|_2^2
+
\mathrm{tr}\!\left(
\Sigma_{1,n}
+
\Sigma_{2,n}
-
2\left(
\Sigma_{2,n}^{\frac{1}{2}}
\Sigma_{1,n}
\Sigma_{2,n}^{\frac{1}{2}}
\right)^{\frac{1}{2}}
\right).
\end{align}

\subsection{Measuring Sensitivity to Learning Rate Estimation Methods}\label{appsubsec:gen-bayes-ising}

We now consider generalised Bayesian inference for the Ising model following \Cref{subsec:gen-bayes-ising}. The likelihood is $p_\theta(x) \propto \exp(\frac{1}{2 \theta} \sum_{i=1}^{\datadim} \sum_{j \in \mathcal{N}_i} x_i x_j)$, where $\mathcal{N}_i$ denotes the neighbours of node $i$ in some undirected graph $G$. The conditional distribution of node $i$ given its neighbours has mass function:
$
p_\theta(x_{r,i} \mid x_{r,j} : j \in \mathcal{N}_i )
=
\exp(\frac{x_{r,i} \, m_i(x_r)}{\theta})/
(1 + \exp(\frac{m_i(x_r)}{\theta})),
$ 
where $r = 1,\ldots,n$ indexes the datapoint, $i = 1,\ldots,d_\mathcal{X}$ 
indexes the node on the grid, $j$ ranges over the neighbours 
$\mathcal{N}_i$ of node $i$ and $m_i(x_r) := \sum_{j \in \mathcal{N}_i} x_{r,j}$.  The PL loss is given by
\begin{align}
\begin{aligned}
 l^\mathrm{PL}(\theta;x_{1:n}) & = - \frac{1}{n}\sum_{r=1}^n \sum_{i=1}^{\datadim}
\log p_\theta(x_{r,i} \mid x_{r,j} : j \in \mathcal{N}_i)\\
& = 
\frac{1}{n}\sum_{r=1}^n\sum_{i=1}^{d_\mathcal{X}}
\left[
\log\left(1+\exp \left(\frac{m_i(x_r)}{\theta}\right)\right)
-
x_{r,i}\frac{m_i(x_r)}{\theta}
\right].
\end{aligned}
\end{align}
and the DFD loss is
\begin{equation}
\begin{aligned}
 l^\mathrm{DFD}(\theta;x_{1:n}) & = 
 \frac{1}{n}\sum_{r=1}^n \sum_{i=1}^{\datadim}
\left[
\left(\frac{p_\theta(x_r^{i})}{p_\theta(x_r)}\right)^2
-2 \frac{p_\theta(x_r)}{p_\theta(x_r^{i})}
\right] \\
& = 
 \frac{1}{n}\sum_{r=1}^n \sum_{i=1}^{\datadim}
\left[
\exp\left(\frac{\Delta_i H(x_r)}{\theta}\right)
-2 \exp\left(-\frac{\Delta_i H(x_r)}{2\theta}\right)
\right].     
\end{aligned}
\end{equation}
Since $x \in \{0,1\}^{\datadim}$, each coordinate $i$ has exactly one available flip; we write $x_r^i$ for the state obtained by flipping the $i$-th coordinate of $x_r$, and define the energy difference $\Delta_i H(x) := H(x^i) - H(x)$ for $H(x) := \sum_{i=1}^{\datadim} \sum_{j \in \mathcal{N}_i} x_i x_j$.

\paragraph{\Cref{assump:fd}.}
Firstly, we verify the support condition for  $\trueposterior$. The parameter space is $\Theta=(0,\infty)$, which is an open connected subset of $\mathbb{R}$, and $\trueprior=\chi^2(3)$, which has full support on $\Theta$. For the PL loss, each 
$p_\theta \left(x_{r,i} \mid x_{r,j} : j \in \mathcal{N}_i \right)$ is bounded above by $1$ for every $\theta \in \Theta$, so each $\log p_\theta \left(x_{r,i} \mid x_{r,j} : j \in \mathcal{N}_i \right) < \infty$ and hence $l^{\mathrm{PL}}(\theta;x_{1:n})<\infty$. For the DFD loss, each term in the sum is finite for every $\theta \in \Theta$ since $\Delta_i H(x_r)$ is bounded on the finite state space $\{ 0,1 \}^{\datadim}$. Therefore, $l^{\mathrm{DFD}}(\theta;x_{1:n})<\infty$. Thus, $\exp(-L(\theta;x_{1:n}, \lambda_\mathrm{ref})) > 0$ for every $\theta \in \Theta$, where $L(\theta;x_{1:n}, \lambda_\mathrm{ref}) = \lambda_{L,\mathrm{ref}} l(\theta;x_{1:n})$ for $l \in 
\{l^\mathrm{PL}, l^\mathrm{DFD}\}$, resulting in $\mathrm{supp}(\tilde{\pi}_\mathrm{ref}) = \Theta$.

Secondly, we verify differentiability and integrability assumptions. Using \Cref{appsec:div_derivation_quadratic_form}, it suffices to verify these conditions for $L$. First, we verify the differentiability condition. For the PL loss,
Each $p_\theta \left(x_{r,i} \mid x_{r,j} : j \in \mathcal{N}_i \right)$ is obtained as a ratio of strictly positive smooth functions of $\theta$ on $\Theta$. Hence, it is $C^1(\Theta)$, and so is its logarithm. Since $l^{\mathrm{PL}}(\theta;x_{1:n})$ is a finite sum of such terms, it follows that $l^{\mathrm{PL}}(\cdot;x_{1:n}) \in C^1(\Theta)$.
For the DFD loss,
each summand of $l^{\mathrm{DFD}}(\theta;x_{1:n})$, the ratios of $p_\theta(\cdot)$ are well-defined and are quotients of $C^1(\Theta)$ functions, as $p_\theta(x)>0$ for every fixed $x$. Therefore, each term is $C^1(\Theta)$, and the finite sum
$l^{\mathrm{DFD}}(\cdot;x_{1:n}) \in C^1(\Theta)$.
Hence, $L^\mathrm{PL}(\cdot;x_{1:n}, \lambda_\mathrm{ref}), \;
L^\mathrm{DFD}(\cdot;x_{1:n}, \lambda_\mathrm{ref}) \in C^1(\Theta)$. 

Now, we check the integrability conditions. For the PL loss, differentiating gives
\begin{equation}
    \nabla_\theta l^{\mathrm{PL}}(\theta;x_{1:n})
    = -
    \frac{1}{n\theta^2}
    \sum_{r=1}^n \sum_{i=1}^{\datadim} m_i(x_r)
    \left[
        x_{r,i}
        -
        \sigma\left(
            \frac{m_i(x_r)}{\theta}
        \right)
    \right],
    \qquad
    \left|
        \nabla_\theta l^{\mathrm{PL}}(\theta;x_{1:n})
    \right|
    \leq \frac{C}{\theta^{2}},
    \label{eq:PL-gradient-bound}
\end{equation}
where $\sigma(t)=e^t/(1+e^t)$. We now split the proof into two cases depending on the observed dataset. In the first case,
$x_{r_0,i_0}=0$ and $m_{i_0}(x_{r_0})>0$ for some
$(r_0,i_0)$, so non-negativity of the PL summands and
$\log(1+e^t)\geq t$ imply
$l^{\mathrm{PL}}(\theta;x_{1:n})\geq c/\theta$ for some
$c>0$. Otherwise, every non-zero gradient term has
$x_{r,i}=1$ and $m_i(x_r)\geq 1$, which gives
$\lvert\nabla_\theta l^{\mathrm{PL}}(\theta;x_{1:n})\rvert
\leq C\theta^{-2}\exp(-\theta^{-1})$.
Using the $\chi^2(3)$ prior density and
$l^{\mathrm{PL}}\geq 0$, respectively, these two cases yield
\begin{equation}
    \mathbb{E}_{\theta\sim\widetilde{\Pi}_{\mathrm{ref}}}
    \left[
        \left|
            \nabla_\theta l^{\mathrm{PL}}(\theta;x_{1:n})
        \right|^2
    \right]
    \lesssim
    \begin{cases}
        \displaystyle
        \int_0^\infty
        \theta^{-\frac{7}{2}}
        \exp\left(
            -\frac{\lambda_{\mathrm{ref}}c}{\theta}
            -\frac{\theta}{2}
        \right)
        \,\mathrm{d}\theta,
        \\[3mm]
        \displaystyle
        \int_0^\infty
        \theta^{-\frac{7}{2}}
        \exp\left(
            -\frac{2}{\theta}
            -\frac{\theta}{2}
        \right)
        \,\mathrm{d}\theta,
    \end{cases}
    <\infty.
\end{equation}
Hence,
$\nabla_\theta l^{\mathrm{PL}}
\in L^2(\widetilde{\Pi}_{\mathrm{ref}})$
for every observed dataset.

We now consider integrability conditions for the DFD loss.
In this case, we impose the following sufficient condition on the observed dataset, which can be verified by a single pass through the data: for the observed dataset $x_{1:n}$, the largest number of $1-$valued neighbours of any $0-$valued node is more than half the largest number of $1-$valued neighbours of any $1-$valued node. That is, $2\max\{m_i(x_r):x_{r,i} = 0\} > \max\{m_i(x_r): x_{r,i} =1\}$ where the maxima are taken over all $r \in \{1,\dots,n\}$ and $i \in \{1,\dots, d_{\mathcal{X}}\}$ and with the convention $\max \emptyset = 0$ for the neighbour counts.
In particular, for our $6\times6$ grid, this condition is satisfied whenever at least one observed configuration contains a $0-$valued node with at least three $1-$valued neighbours.

With this condition in place, we can proceed to show that $\nabla_\theta l^{\mathrm{DFD}} \in L^2(\widetilde{\Pi}_{\mathrm{ref}})$. It suffices to show that the reference posterior is well-defined, i.e. $0< Z_{\text{ref}} := \int_0^\infty
\exp\{-\lambda_{\mathrm{ref}}l^{\mathrm{DFD}}(\theta;x_{1:n})\}
\pi_{\mathrm{ref}}(\theta)\,\mathrm{d}\theta
< \infty$ and that $\int_0^\infty I(\theta)\,d\theta<\infty$, where, using $\pi_{\mathrm{ref}}(\theta)\propto \theta^{1/2}e^{-\theta/2}$ for the $\chi^2(3)$ density,
\begin{align} \label{eq:i}
I(\theta)
:=
\left|
\nabla_\theta l^{\mathrm{DFD}}(\theta;x_{1:n})
\right|^2
\exp\left\{
-\lambda_{\mathrm{ref}}
l^{\mathrm{DFD}}(\theta;x_{1:n})
\right\}
\theta^{1/2}e^{-\theta/2}.
\end{align}
We first study $I(\theta)$; the bounds obtained below will also imply that $l^{\text{DFD}}$ is bounded from below, and hence that $Z_{\text{ref}} < \infty$.

The function $I$ is continuous on $(0,\infty)$, since
$l^{\mathrm{DFD}}(\cdot; x_{1:n})\in C^1((0,\infty))$. It is therefore
integrable on every compact subset of $(0,\infty)$. Thus it
suffices to establish integrability at the two boundaries,
$\theta\to0$ and $\theta\to\infty$. 

To establish that the integral does not diverge as $\theta \to \infty$, we bound each component of $I(\theta)$ separately on $[1,\infty)$, since $\int_{0}^\infty I(\theta)d\theta = \int_{1}^\infty I(\theta)d\theta + \int_{0}^1 I(\theta)d\theta$. First, note that 
    \begin{align}
    \left| \nabla_\theta l^{\text{DFD}}(\theta; x_{1:n}) \right|^2 
    &= \left|-
    \frac{1}{n}\sum_{r=1}^n\sum_{i=1}^{\datadim}
    \left[
    \frac{\Delta_i H(x_r)}{\theta^2}
    \exp\left(\frac{\Delta_i H(x_r)}{\theta}\right)
    +
    \frac{\Delta_i H(x_r)}{\theta^2}
    \exp\left(-\frac{\Delta_i H(x_r)}{2\theta}\right)
    \right]
    \right|^2 \nonumber \\
    &\leq \frac{C}{\theta^4} \exp\left(\frac{B}{\theta} \right)
    \end{align}
 for $B = 4\max_{r,i}|\Delta_i H(x_r)| < \infty$ and $C < \infty$ 
    absorbs constants from the finite sum. For $\theta \geq 1$, $\exp(B/\theta) \leq \exp(B)$ and hence $|\nabla_\theta l^{\text{DFD}}(\theta;x_{1:n})|^2 \leq C \exp(B) \theta^{-4}$. We now turn to the next term. 
    For $\theta\geq1$, $\exp(-\frac{\Delta_iH(x_r)}{2\theta})
\leq \exp (\max_{r,i}|\Delta_iH(x_r)|)$. 
Since the first term in each DFD summand is non-negative, it follows that 
\begin{align} \label{eq:lower-bound-big-theta}
l^{\mathrm{DFD}}(\theta;x_{1:n})
\geq
-\frac{2}{n}\sum_{r=1}^n\sum_{i=1}^{\datadim}\exp (\max_{r,i}|\Delta_iH(x_r)|)
=
-2\datadim \exp (\max_{r,i}|\Delta_iH(x_r)|).
\end{align}
Hence $l^{\mathrm{DFD}}$ is bounded from below on $[1,\infty)$ and $\exists C_1<\infty$ such that $\exp(-\lambda_{\mathrm{ref}}
l^{\mathrm{DFD}}(\theta;x_{1:n}))\leq C_1$. Therefore, 
substituting these bounds into (\ref{eq:i}) gives $I(\theta) \le
CC_1 \exp(B)\theta^{-7/2}e^{-\theta/2}$ for $ \theta\ge1$ and it follows that
\begin{align} \label{eq:ithetamorethanone}
\int_1^\infty I(\theta)\,d\theta<\infty.
\end{align}

We now examine the behaviour as $\theta\to0$.
Each term in the DFD loss is of the form $c \exp(\frac{a}{\theta})$, where $a$ determines whether the term decays or grows as $\theta \to 0$. We refer to $a$ as its exponential rate. Since the dataset is finite, there are only finitely many distinct rates. Collecting terms with the same rate, we may therefore write:
\begin{align}  \label{eq:dfd-form}
l^{\mathrm{DFD}}(\theta;x_{1:n})
=
\sum_{a\in\mathcal A} c_a e^{a/\theta},
§\end{align}
where $\mathcal A\subset\mathbb R$ is the finite set of exponential rates having non-zero net coefficient $c_a$ with
$
c_a
:=\frac{1}{n}\left[\#\{(r,i):\Delta_i H(x_r)=a\}-2\#\{(r,i):-\frac{\Delta_i H(x_r)}{2}=a\}\right].
$
Differentiating gives
\begin{align} \label{eq:dfd-grad}
\nabla_\theta l^{\mathrm{DFD}}(\theta;x_{1:n})
=
-\theta^{-2}
\sum_{a\in\mathcal A}
a c_a e^{a/\theta}.
\end{align}
Notice there are three scenarios for the exponential term in (\ref{eq:dfd-grad}): (i) for $a<0$, $e^{a/\theta}\to0$ as $\theta\to0$, (ii) for $a=0$, $e^{a/\theta}=1$, and (iii) for $a>0$, it diverges.
We therefore focus on the exponentially growing terms. Define
$ \mathcal A_+ := \{a\in\mathcal A:a>0\}$, the set collecting all the positive exponential rates.
If $\mathcal{A}_{+} = \varnothing$, all exponential terms appearing in the gradient are non-growing as $\theta \to 0$, which would make controlling integrability straightforward. However, this case is ruled out by our dataset condition, which guarantees the existence of a positive exponential rate. We must therefore consider the more challenging case $\mathcal{A}_{+} \neq \varnothing$, in which the gradient may grow exponentially as $\theta \to 0$. Crucially, as we show below, the dataset condition also ensures that the largest such rate arises from a positive contribution to the loss, so that its coefficient is positive. The resulting growth of the loss then produces sufficient decay in the posterior density to dominate the growth of the gradient.

The key point is that exponential growth plays two opposing roles in the integrand $I(\theta)$. Growth of $|\nabla_\theta l^{\mathrm{DFD}}(\theta;x_{1:n})|$ can be detrimental to integrability, whereas growth of $l^{\mathrm{DFD}}(\theta;x_{1:n})$ to $+\infty$ is beneficial, since the loss enters through the factor
$\exp\{-\lambda_{\mathrm{ref}}l^{\mathrm{DFD}}(\theta;x_{1:n})\}$.
We therefore upper-bound the growth of the gradient while lower-bounding the growth of the loss.
As shown below, both are governed by the largest positive exponential rate, which we denote as  $a_\star:=\max_{a\in\mathcal A_+} a>0$, but the posterior factor induced by the growth of the loss decays double-exponentially, which dominates both the exponential and polynomial growth appearing in the gradient.

We now show how the dataset condition ensures that the largest exponential rate has a positive coefficient. Recall that for a $0$-valued node we have $\Delta_i H(x_r) = m_i(x_r)$, so that positive contributions to the loss give exponential rates $m_i(x_r)$. In contrast, for a $1$-valued node we
have $\Delta_i H(x_r)=-m_i(x_r)$, so that the corresponding negative contributions to the loss give positive exponential rates
$m_i(x_r)/2$. Hence, the dataset condition
$2\max\{m_i(x_r):x_{r,i}=0\} > \max\{m_i(x_r):x_{r,i}=1\}$ ensures that the largest positive exponential rate arising from a positive
contribution is strictly larger than the largest positive exponential rate arising from a negative contribution. Therefore, $a_\star$ is attained by at least one positive contribution of the form $\Delta_i H(x_r)$ and is strictly larger than every positive rate arising from a negative contribution of the form $-\frac{\Delta_i H(x_r)}{2}$. Therefore, $\#\{(r,i): - \frac{\Delta_i H(x_r)}{2} = a_\star\} = 0$ and since $a_\star$ is attained, we have $\#\{(r,i): \Delta_iH(x_r) = a_\star \} \geq 1$. It therefore follows from the definition of $c_a$ that 
$c_{a^\star} = \frac{1}{n} \#\{(r,i):  \Delta_iH(x_r) = a_\star\} > 0$.

Separating the term corresponding to
the largest rate gives
\begin{align}
l^{\mathrm{DFD}}(\theta;x_{1:n})
=
c_{a_\star}e^{a_\star/\theta}
+
\sum_{a<a_\star}c_a e^{a/\theta}
=
e^{a_\star/\theta}
\left[
c_{a_\star}
+
\sum_{a<a_\star}
c_a e^{(a-a_\star)/\theta}
\right].
\label{eq:dfd-leading-rate}
\end{align}

Now for every $a<a_\star$, we have that $e^{(a -a_\star)/\theta} \to 0$ as $\theta\to0$.
Since the sum in \eqref{eq:dfd-leading-rate} is finite, the expression in brackets then converges to $c_{a_\star}$ as $\theta \to 0$. 
Hence, there exists $\epsilon>0$ such that
\begin{align}
l^{\mathrm{DFD}}(\theta;x_{1:n})
&\ge
\frac{c_{a_\star}}{2}e^{a_\star/\theta},
\qquad 0<\theta<\epsilon.
\label{eq:dfd-loss-lower-bound}
\end{align}
Moreover, since $a_\star$ is the largest exponential rate and
$\mathcal A$ is finite, it follows from \eqref{eq:dfd-grad} that there exists $C_5<\infty$ such that:
\begin{align}
\left|
\nabla_\theta l^{\mathrm{DFD}}(\theta;x_{1:n})
\right|
&\le
C_5\theta^{-2}e^{a_\star/\theta}.
\label{eq:dfd-gradient-upper-bound}
\end{align}
Substituting \eqref{eq:dfd-loss-lower-bound} and
\eqref{eq:dfd-gradient-upper-bound} into the target integrand
$I(\theta)$ and pulling finite constants together into $C_6 < \infty$, yields
\begin{align}
I(\theta)
&\le
C_6\theta^{-7/2}
\exp\left\{
\frac{2a_\star}{\theta}
-
\frac{\lambda_{\mathrm{ref}}c_{a_\star}}{2}
e^{a_\star/\theta}
-
\frac{\theta}{2}
\right\},
\qquad 0<\theta<\epsilon.
\label{eq:dfd-case2-integrand-bound}
\end{align}
As $\theta\to0$, the negative double-exponential term
$-\frac{\lambda_{\mathrm{ref}}c_{a_\star}}{2}
e^{a_\star/\theta}$
dominates both the growth arising from $2a_\star/\theta$ and the
polynomial factor $\theta^{-7/2}$. Therefore, the right-hand side of
\eqref{eq:dfd-case2-integrand-bound} is integrable at zero, and hence $\int_0^{\epsilon} I(\theta)\,d\theta<\infty$.
Together with \eqref{eq:ithetamorethanone} and the
continuity of $I$ on $(0,\infty)$, this gives $\int_0^\infty I(\theta)\,d\theta<\infty$. 
It remains to verify that $Z_{\text{ref}} < \infty$.
Indeed, \eqref{eq:dfd-loss-lower-bound} provides a lower bound for sufficiently small $\theta$, \eqref{eq:lower-bound-big-theta} provides a lower bound for $\theta \geq 1$, and continuity provides a lower bound on the remaining compact interval. Hence, $\inf_{\theta>0} l^{\mathrm{DFD}}(\theta;x_{1:n}) > -\infty$. Using that $\pi_{\mathrm{ref}}$ is a proper prior,
\[
0 < Z_{\mathrm{ref}}
\leq
e^{-\lambda_{\mathrm{ref}}\inf_{\theta>0} l^{\mathrm{DFD}}(\theta;x_{1:n})}
\int_0^\infty \pi_{\mathrm{ref}}(\theta)\,\mathrm{d}\theta
=
e^{-\lambda_{\mathrm{ref}}\inf_{\theta>0} l^{\mathrm{DFD}}(\theta;x_{1:n})}
<\infty.
\]
Thus, $\widetilde{\Pi}_{\mathrm{ref}}$ is well-defined and together with $\int_0^\infty I(\theta)\,d\theta<\infty$, it follows that
$
\nabla_\theta
l^{\mathrm{DFD}}
\in
L^2(\widetilde{\Pi}_{\mathrm{ref}}).
$

Next, we verify that
$\mathcal{P}_{\Gamma}$
is a subset of $\mathcal{P}_\FD(\Theta)$.
The $C^1(\Theta)$ regularity of each candidate loss follows immediately since 
the reference losses are in $C^1(\Theta)$. For the integrability condition, we have for each loss $l \in 
\{l^\mathrm{PL}, l^\mathrm{DFD}\}$,
\begin{align}
\begin{aligned}
\mathbb{E}_{\theta \sim \trueposterior} \left[ |\nabla_\theta L(\theta;x_{1:n},\lambda_L)|^2 \right]
& = \lambda_L^2 \mathbb{E}_{\theta \sim \trueposterior} \left[ |\nabla_\theta l(\theta;x_{1:n})|^2 \right]\\
&\leq (\lambda_{L,\mathrm{ref}} + \epsilon)^2 \mathbb{E}_{\theta \sim \trueposterior} \left[ |\nabla_\theta l(\theta;x_{1:n})|^2 \right]< \infty
\end{aligned}
\end{align}
where the last integral is finite since it was already verified for 
the reference losses above. Therefore,
$\nabla_\theta L^\mathrm{PL},\; 
\nabla_\theta L^\mathrm{DFD} \in L^2(\trueposterior)$
uniformly over $\lambda_L \in \Gamma$.
Finally, since $\candidateprior = \trueprior$ and each loss is finite 
on $(0,\infty)$ for all $\lambda_L \in \Gamma$, the support of each 
candidate posterior satisfies 
$\mathrm{supp}(\tilde{\pi}^\lambda) = \mathrm{supp}(\tilde{\pi}_{\mathrm{ref}})
= (0,\infty)$
for all $\lambda_L \in \Gamma$. 

\paragraph{\Cref{assump:minimal-exp-fam}.}
We are doing learning rate sensitivity and hence the assumption is satisfied.

\subsection{Kilpisjärvi temperature autoregressive model}\label{appsubsec:heavy-tailed-ar-model-kilpisjarvi}
In this experiment, sensitivity is analysed only with respect to the prior, so that $\lambda=\lambda_\pi$.

\paragraph{\Cref{assump:fd}.} We will now show that this assumption holds. Define the design matrix of lagged predictors as: 
\begin{align*}
    Z = \begin{bmatrix}1 & x_S & x_{S-1} &\dots & x_1 \\
    1 & x_{S+1} & x_{S} &\dots & x_2 \\
    \vdots & \vdots & \vdots & & \vdots \\
    1 & x_{n-1} & x_{n-2} &\dots & x_{n-S} 
    \end{bmatrix}.
\end{align*} Throughout, we assume that the observed response vector \(x=(x_{S+1},\ldots,x_n)^\top\) does not lie in the column space of the lagged design matrix \(Z\), equivalently \(\delta^2:=\inf_{\theta'}\|x-Z\theta'\|_2^2>0\), a condition that can be checked directly from the observed dataset.

We first verify the support condition for $\trueposterior$. Since $\theta\mapsto z$ is a smooth bijection onto $\mathbb{R}^{\paramdim}$, the parameter space $\mathbb{R}^{\paramdim}$ is open and connected.
The autoregressive loss component is a standard log-likelihood which is finite 
for every $\theta \in \Theta$ and $\exp(-L_{\mathrm{ref}}(\theta;x_{1:n}, 
\lambda_\mathrm{ref}))>0$ for all $\theta \in \Theta$. The reference prior is
a product of prior per parameter, each of which have full support, therefore it also has full support. Since both the prior and generalised update have full support, $\trueposterior$ has full support.

We now verify the differentiability conditions. The reference prior $\pi_{\mathrm{ref}}(z)=\prod_{j}\phi(z_j)$ is in $C^\infty(\mathbb{R}^{\paramdim})$ with $s_{\pi_{\mathrm{ref}}}(z)=-z$. Writing $\theta'=(\alpha,\beta_1,\dots,\beta_S)$ and $k=n-S$, the likelihood satisfies $\ell(\theta)\le(2\pi\sigma^2)^{-k/2}\exp(-\delta^2/(2\sigma^2))\le M:=(2\pi)^{-k/2}(k/\delta^2)^{k/2}e^{-k/2}<\infty$. Hence, for any $h\ge0$,
\begin{align}\label{appeq:ar-moment}
\mathbb{E}_{z\sim\trueposterior}[h(z)]\le \tfrac{M}{Z_{\mathrm{ref}}}\,\mathbb{E}_{z\sim\mathcal N(0,I_{\paramdim})}[h(z)],
\end{align}
and taking $h(z)=\|z\|_2^2$ gives $\mathbb{E}_{z\sim\trueposterior}[\|s_{\pi_{\mathrm{ref}}}(z)\|_2^2]\le \paramdim M/Z_{\mathrm{ref}}<\infty$. Thus, $s_{\pi_{\mathrm{ref}}}\in L^2(\tilde{\Pi}_{\mathrm{ref}})$.

Next, we verify that the posterior neighbourhood $\mathcal{P}_{\Gamma}$ is a subset of $\mathcal{P}_{\FD}(\Theta)$. We begin with the support condition, and first note that $\pi(\theta|\lambda)>0$ is strictly positive for every $\theta \in \Theta$ and $\lambda \in \Gamma$ and the likelihood is the same for the reference and candidate posteriors. Therefore, the candidate posterior density $\tilde{\pi}(\cdot|\lambda)$ has support
$\mathrm{supp}(\tilde{\pi}(\cdot|\lambda))
= \Theta =
\mathrm{supp}(\tilde{\pi}_{\mathrm{ref}})$. For the differentiability condition, we note that $\pi(\theta|\lambda) \in C^1(\Theta)$ as it is a composition of continuously differentiable functions.
For integrability, $[s_{\pi(\cdot|\lambda)}(z)]_j=-(z_j-\mu_j)/s_j^2$, so $|\mu_j|\le0.4$ and $s_j\ge0.8$ give $\|s_{\pi(\cdot|\lambda)}(z)\|_2^2\le 2\cdot0.8^{-4}(\|z\|_2^2+0.4^2\paramdim)$ uniformly over $\Gamma$. By \Cref{appeq:ar-moment}, $s_{\pi(\cdot|\lambda)}\in L^2(\tilde{\Pi}_{\mathrm{ref}})$ for all $\lambda\in\Gamma$, and therefore $\mathcal{P}_{\Gamma}\subset \mathcal{P}_{\FD}(\mathbb{R}^{\paramdim})$.

\paragraph{\Cref{assump:minimal-exp-fam}.}
Each component of $\candidateprior$ is in natural exponential-family form with $\lambda_j=[\frac{\mu_j}{s_j^2},-\frac{1}{2s_j^2}]^\top$ and $T_j(z_j)=[z_j,z_j^2]^\top$. Collecting all components gives
$\lambda=[\lambda_\alpha,\lambda_{\beta_1},\dots,\lambda_{\beta_S},\lambda_\sigma]^\top$ and $T(z)=[z_\alpha,z_\alpha^2,z_{\beta_1},z_{\beta_1}^2,\dots,z_{\beta_S},z_{\beta_S}^2,z_\sigma,z_\sigma^2]^\top$.

\subsection{Simulation-based inference for radio propagation models}\label{appsubesc:sbi}

\paragraph{\Cref{assump:fd}.} We will now show that this assumption fails in the original parametrisation, but holds in the alternative parametrisation.

We first verify the support condition for $\trueposterior$. The parameter space $\Theta=\prod_{j=1}^{\paramdim}(a_j,b_j)
$ is an open connected subset of $\mathbb R^{\paramdim}$. Since both
$q_{\hat{\phi}}(x_{1:n}\mid\theta)$ and $\pi_{\mathrm{ref}}(\theta)$
are strictly positive on $\Theta$, so is $\tilde{\pi}_{\mathrm{ref}}$.
Hence, $\operatorname{supp}(\tilde{\pi}_{\mathrm{ref}})=\Theta
$.

Since we are performing prior sensitivity, we only need to check the differentiability and integrability conditions for the priors (see \Cref{appsec:div_derivation_quadratic_form}).
Starting with the reference, $\pi_{\mathrm{ref}}(\theta)
=\prod_{j=1}^{\paramdim} \pi_{\mathrm{ref},j}(\theta_j)$ for
$\pi_{\mathrm{ref},j}(\theta_j)= (b_j - a_j)^{-1}$, and $\pi_{\mathrm{ref}}(\theta)$ is therefore a strictly positive constant. Therefore, $\pi_{\mathrm{ref}} \in C^1(\Theta)$, $s_{\pi_{\mathrm{ref}}}(\theta)=0$ and $\mathbb{E}_{\theta \sim \trueposterior}[
\|s_{\pi_{\mathrm{ref}}}(\theta)\|_2^2]
= 0 < \infty$. 

Next, we consider the candidate priors
$\pi(\theta|\lambda) = \pi_\mathrm{ref}(\theta) c_{\lambda}(\theta_i, \theta_j)$ for $\lambda \in (-1,1) \backslash \{0\}$.
The candidate posterior density satisfies
$\tilde{\pi}^\lambda(\theta) >0$
for all $\theta \in \Theta$, since the reference prior density, the copula density
and the neural likelihood estimator are strictly positive on $\Theta$.
Hence $\mathrm{supp}(\tilde{\pi}^\lambda) = \Theta = \mathrm{supp}(\tilde{\pi}_\mathrm{ref})$.
Since $\pi_{\mathrm{ref}} \in C^1(\Theta)$ and the copula density
is a composition of $C^1(\Theta)$ functions on $\Theta$, we have $\pi(\cdot|\lambda) \in C^1(\Theta)$ from the product and chain rule.

It remains to verify the integrability condition.
The score of the candidate prior reduces to the score of the copula: $s_{\pi(\cdot|\lambda)}(\theta)
= s_{\pi_\mathrm{ref}}(\theta) + \nabla_\theta \log c_{\lambda}(\theta_i, \theta_j)
= \nabla_\theta \log c_{\lambda}(\theta_i, \theta_j)$ where $\theta_i,\theta_j$ are the parameters we want to correlate. Recall that $F_i(\theta_i) = (\theta_i - a_i)/(b_i - a_i)$ is the CDF of the $i$-th 
reference prior marginal and $z_i = \Phi^{-1}(F_i(\theta_i))$.
Since the copula density depends on $\theta$ only through
$\theta_i,\theta_j$, all components of
$\nabla_\theta\log c_\lambda(\theta_i,\theta_j)$ other than the
$i$-th and $j$-th are zero
\begin{align}
\frac{\partial}{\partial \theta_i} \log c_{\lambda}(\theta_i, \theta_j)
=
\frac{\lambda z_j - \lambda^2 z_i}{(1-\lambda^2)\,\phi(z_i)\,(b_i - a_i)},
\quad
\frac{\partial}{\partial \theta_j} \log c_{\lambda}(\theta_i, \theta_j)
=
\frac{\lambda z_i - \lambda^2 z_j}{(1-\lambda^2)\,\phi(z_j)\,(b_j - a_j)},
\end{align}
where $\phi$ denotes the standard Gaussian density. 
The squared norm is
\begin{equation}\label{appeq:full-squared-norm}
\|s_{\pi(\cdot\mid\lambda)}(\theta)\|_2^2
=
\underbrace{\left(\frac{\lambda z_j-\lambda^2 z_i}{(1-\lambda^2)\phi(z_i)(b_i-a_i)}\right)^{\!2}}_{=:A_i(\theta)}
+
\underbrace{\left(\frac{\lambda z_i-\lambda^2 z_j}{(1-\lambda^2)\phi(z_j)(b_j-a_j)}\right)^{\!2}}_{=:A_j(\theta)}.
\end{equation}
Here, $A_i(\theta), \;A_j(\theta) \ge 0$ for every $\theta$, so 
$\|s_{\pi(\cdot\mid\lambda)}(\theta)\|_2^2 \ge A_i(\theta)$.
Now, we define a sub-region of $\Theta$ on which $A_i(\theta)$ admits an explicit lower bound. For $C>0$ and $\theta_{-i}=(\theta_k)_{k\neq i}$ define
$K \;:=\; \bigl\{\theta_{-i} \in \textstyle\prod_{k\neq i}(a_k,b_k) \,:\, |z_j| \le C \bigr\}$. Since $z_j$ is a continuous, strictly
increasing bijection of $\theta_j\in(a_j,b_j)$ onto $\mathbb{R}$, the condition $|z_j|\le C$
corresponds to a subinterval of $(a_j,b_j)$; the remaining coordinates
$\theta_k$, $k\neq i,j$, range freely over bounded intervals $(a_k,b_k)$. Hence $K$
has strictly positive and finite Lebesgue measure, $0<\mathrm{Leb}(K)<\infty$.
For $\theta_i$, define the threshold
$R := a_i + (b_i-a_i)\Phi\left(2C / |\lambda|\right)$,
so that $\theta_i>R \iff z_i > 2C/|\lambda|$. For $\theta_{-i}\in K$ and $\theta_i>R$, the
reverse triangle inequality and $|z_j|\le C < |\lambda||z_i|/2$ give
$|\lambda z_j - \lambda^2 z_i| \ge (\lambda^2 / 2)|z_i|$, leading to the following lower bound
\begin{equation}\label{appeq:pointwise-lower-bound}
\|s_{\pi(\cdot\mid\lambda)}(\theta)\|_2^2 \;\ge\; A_i(\theta)
\;\ge\;
\frac{\lambda^4}{4(1-\lambda^2)^2(b_i-a_i)^2}\,\frac{z_i^2}{\phi(z_i)^2}.
\end{equation}
We also define the neural likelihood estimator $q_{\hat{\phi}}(x_{1:n}\mid\theta)$ as a masked  autoregressive flow (MAF) \citep{papamakarios2017masked}
\begin{align*}
q_{\hat{\phi}}(x_{1:n} \mid \theta) 
= (2\pi)^{-d_{\mathcal{X}}/2} 
\exp\!\left(-\frac{1}{2}\|u(x;\theta)\|^2\right) 
\prod_{r=1}^{d_{\mathcal{X}}} \sigma_{\hat{\phi},r}(x_{<r},\theta)^{-1},
\end{align*}
where $u:=(u_1,\dots,u_{\datadim})$, $u_r(x;\theta) = \bigl(x_r - \mu_{\hat{\phi},r}(x_{<r},\theta)\bigr)/\sigma_{\hat{\phi},r}(x_{<r},\theta)$,
the shift $\mu_{\hat{\phi},r}$ and pre-activation scale 
$s_{\hat{\phi},r}$ networks are affine functions of $(x_{<r}, \theta)$ with $\tanh$ activations, and $\sigma_{\hat{\phi},r}(x_{<r},\theta) = \mathrm{softplus}(s_{\hat{\phi},r}(x_{<r},\theta))$.
$q_{\hat\phi}(x_{1:n}\mid\theta)$ is a continuous, strictly positive function on
$\overline{\Theta}$ (compact closure of $\Theta$), and therefore there exist constants $0<m\le M<\infty$ such that
$m \;\le\; q_{\hat\phi}(x_{1:n}\mid\theta) \;\le\; M$ for $\theta \in \Theta$. Restricting the domain of integration to $\{\theta_i > R\}\times K \subset \Theta$ and using the fact that $\pi_{\mathrm{ref}}$ is constant on $\Theta$, we obtain
\begin{align}
\mathbb{E}_{\theta\sim\trueposterior}\bigl[\|s_{\pi(\cdot\mid\lambda)}(\theta)\|_2^2\bigr]
&\;\ge\;
\int_{\theta_i>R}\int_{K}
\frac{\lambda^4}{4(1-\lambda^2)^2(b_i-a_i)^2}\,\frac{z_i^2}{\phi(z_i)^2}\,Z_{\text{ref}}^{-1}
\pi_{\mathrm{ref}}(\theta)\,q_{\hat\phi}(x_{1:n}\mid\theta)\;d\theta_{-i}\,d\theta_i \notag \\
&\;\ge\; C_\lambda \int_R^{b_i} \frac{z_i^2}{\phi(z_i)^2}\,d\theta_i \overset{(1)}{=} 
C_\lambda (b_i-a_i)\int_{2C/|\lambda|}^{\infty} \frac{z_i^2}{\phi(z_i)}\,dz_i \overset{(2)}{=} \infty,
\end{align}
where $C_\lambda \;:= m \pi_{\mathrm{ref}} \mathrm{Leb}(K)\lambda^4 / (4(1-\lambda^2)^2(b_i-a_i)^2)Z_{\text{ref}}^{-1} >0$ and $Z_{\text{ref}}$ is the normalisation constant of the reference posterior. $(1)$ holds due to the change of variables and $(2)$ holds as $\phi(z_i)$ decays faster than any polynomial growth as $z_i\to\infty$. 
Thus, for every \(\lambda\in(-1,1)\setminus\{0\}\),
\(s_{\pi(\cdot\mid\lambda)}\notin L^2(\widetilde\Pi_{\mathrm{ref}})\) and, in fact,
\(\mathrm{FD}(\widetilde\Pi_{\mathrm{ref}}\Vert\widetilde\Pi^\lambda)=\infty\). Hence \Cref{assump:fd} is violated in the original parametrization.

For \Cref{assump:fd} to hold, we reparametrise the model in terms of $z$ and write the new posterior as $\trueposterior^z$.
The reference prior becomes 
$\pi^z_{\mathrm{ref}}(z) = \mathcal{N}(z;0, I_{\paramdim \times \paramdim})$ 
and the candidate prior is 
$\pi^z(z|\lambda) = \pi^z_\mathrm{ref}(z)\, c_{\lambda}(z_i, z_j) 
= \mathcal{N}(z;0, \Sigma)$, 
where $\Sigma$ has $1$-s on the main diagonal and $\lambda$ in the $(i,j)$-th 
off-diagonal coordinates. The score of the candidate prior is
$s_{\pi^z(\cdot\mid \lambda)}(z) = -\Sigma^{-1} z$,
which is affine in $z$. Therefore, $\mathbb{E}_{z \sim \trueposterior^z}\left[\|s_{\pi^z(\cdot\mid \lambda)}(z)\|_2^2\right]
\leq \|\Sigma^{-1}\|_{\mathrm{op}}^2\, \mathbb{E}_{z \sim \trueposterior^z}\left[\|z\|_2^2\right]$. $\|\Sigma^{-1}\|_{\mathrm{op}} < \infty$ for $|\lambda| < 1$, and we have
$\mathbb{E}_{z\sim\trueposterior^z}\bigl[\|z\|_2^2\bigr]
\le  \paramdim M / Z_{\text{ref}}
\;<\; \infty$,
using $\int\|z\|^2\phi(z)\,dz=\paramdim$ for the standard $\paramdim$-dimensional Gaussian.
Hence, $s_{\pi^z(\cdot\mid \lambda)}(z) \in L^2(\tilde{\Pi}_{\mathrm{ref}}^z)$. 
All the other conditions of \Cref{assump:fd} are satisfied in the reparametrised space as well as we work with standard Gaussians (following similar logic to \Cref{appsubsec:toy_gaussian_model}).

\end{document}